\documentclass[journal]{IEEEtran}
\usepackage{hyperref}

\usepackage{graphicx}      

\usepackage[cmex10]{amsmath}
\usepackage{amsfonts}
\usepackage{color}
\usepackage{psfrag}
\usepackage{subfigure}
\usepackage{empheq}
\usepackage{tabularx}
\usepackage{bm} 
\usepackage{nicefrac} 

\usepackage{url}
\usepackage{psfrag}
\usepackage{subfigure}

\usepackage{empheq}
\usepackage{tabularx}
\usepackage{multirow}
\usepackage{algorithm}
\usepackage{algpseudocode}
\usepackage{booktabs}

\usepackage{enumerate}

\usepackage{cleveref}

\usepackage[EULERGREEK]{sansmath}

\usepackage{pgfplots}
\pgfplotsset{compat=newest}

\usepackage{tikz}
\usepackage{colortbl}
\usepackage{url}
\usepackage{zref-savepos}

\usepackage{makecell}

\usetikzlibrary{external}
\tikzset{external/force remake}

\usepgfplotslibrary{fillbetween}

\newtheorem{thm}{Theorem} 

\newtheorem{lem}{Lemma} 
\newenvironment{proof}{\textit{Proof.}}{~\hfill\rule{0.75em}{0.75em}\\}

\newcommand{\ul}[1]{\underline{#1}}
\newcommand{\ol}[1]{\overline{#1}}

\newcommand{\myvec}[1]{\text{vec}_{#1}}

\newcommand{\R}{\mathbb{R}}
\newcommand{\C}{\mathbb{C}}

\newcommand{\cH}{\mathcal{H}}

\newcommand{\cP}{\mathcal{P}}

\newcommand{\cS}{\mathcal{S}}

\newcommand{\vx}{\mathbf{x}}
\newcommand{\vw}{\mathbf{w}}
\newcommand{\vtw}{\widetilde{\vw}}
\newcommand{\vK}{\mathbf{K}}
\newcommand{\vtK}{\widetilde{\vK}}

\newcommand{\vy}{\mathbf{y}}

\newcommand{\vP}{\mathbf{P}}
\newcommand{\tG}{\widetilde{G}}

\newcommand{\vomega}{\bm{\omega}}

\newcommand{\vQ}{\mathbf{Q}}

\newcommand{\vV}{\mathbf{V}}

\newcommand{\vtheta}{\bm{\theta}}

\newcommand{\vR}{\mathbf{R}}

\newtheorem{rem}{Remark} 

\newcommand{\Hinf}{\cH_{\infty}}
\newcommand{\Htwo}{\cH_2}
\newcommand{\DefinedAs}[0]{\mathrel{\mathop:}=}

\DeclareMathOperator*{\argmin}{arg\,min}

\newcommand{\bigSigma}{\overline{\sigma}}
\newcommand{\jw}{j\omega}

\newcommand{\kit}{{(k)}}
\newcommand{\kitprev}{{(k-1)}}

\newcommand{\blkdiag}[1]{\text{bd}_{#1}}

\newcommand{\HinfNorm}[1]{\left\| #1 \right\|_\infty}

\newlength\fheight
\newlength\fwidth
\newlength\fheightTwo
\newif\ifcommenttorolf
\commenttorolftrue


\usepackage[switch]{lineno}

\begin{document}
	%
	\title{Scalable and Data Privacy Conserving\\ Controller Tuning for Large-Scale Power Networks}
	%
	%
	%
	
	\author{Amer~Me{\v s}anovi{\'c},
		Ulrich~M{\"u}nz,~\IEEEmembership{Member,~IEEE,}
		and~Rolf~Findeisen,~\IEEEmembership{Member,~IEEE,}
		\thanks{RF (rolf.findeisen@ovgu.de) and AM (amer.mesanovic@siemens.com) are  with the Laboratory for Systems Theory and Automatic Control, Otto-von-Guericke-University Magdeburg, Germany, AM is also with the Siemens AG, Munich, UM (ulrich.muenz@siemens.com) is with Siemens Corp., Princeton.}
		\thanks{This work has been partially funded by the German Federal Ministry of Education and Research (BMBF) under Grant number 01S18066B in the frame of the AlgoRes project.}}
	
	%
	%

	\maketitle
	
	\begin{abstract}
		The increasing share of renewable generation leads to new challenges in reliable power system operation, 
		such as the rising volatility of power generation, which leads to time-varying dynamics and behavior of the system.
		To counteract the changing dynamics, we propose to adapt the parameters of existing controllers to the changing conditions. 
		Doing so, however, is challenging, as large power systems often involve multiple subsystem operators, which, for safety and privacy reasons, do not want to exchange detailed information about their subsystems. 
		Furthermore, centralized tuning of structured controllers for large-scale systems, such as power networks, is often computationally very challenging.
		For this reason, we present a hierarchical decentralized approach for controller tuning, which increases data security and scalability. 
		The proposed method is based on the exchange of structured reduced models of subsystems, which conserves data privacy and reduces computational complexity. For this purpose, suitable methods for model reduction and model matching are introduced.
		Furthermore, we demonstrate how increased renewable penetration leads to time-varying dynamics on the IEEE 68 bus power system, which underlines the importance of the problem. 
		Then, we apply the proposed approach on simulation studies to show its effectiveness.
		As shown, a similar system performance as with a centralized method can be obtained.
		Finally, we show the scalability of the approach on a large power system with more than 2500 states and about 1500 controller parameters.
	\end{abstract}
	
	\begin{IEEEkeywords}
		power system, hierarchical optimization, large-scale systems, structured controller synthesis, H-infinity design, linear matrix inequalities, power oscillation damping, data security and integrity
	\end{IEEEkeywords}

	%
	\IEEEpeerreviewmaketitle

	
	\allowdisplaybreaks

	\section{Introduction}
	
	Interconnecting electric power transmission systems has many widely recognized advantages, compared to isolated operation of smaller power systems.
	These include increased system resiliency, i.e. the ability of the system to withstand larger disturbances, cheaper operation, and reduced frequency variations~\cite{kundur93a}. 
	Over the last decades, smaller power systems have been intensively coupled. This has lead to large systems such as the European grid or the Western Interconnection in the USA. Such systems often consist of hundreds of interconnected subystems, managed by different subsystem operators.
	Each subsystem can consist of many different components, such as power plants, wind turbines, households and charging stations, which are interconnected by a power grid, c.f. Fig.~\ref{fig.GridModel}. 
	
	Operating such large systems of different components from different vendors is challenging. 
	Over the decades, a complex decentralized automation system was developed, based on years of practical experience and operation, which enables reliable operation of power systems. Power system components, spanning from inverters to large power plants, are typically operated using, e.g., PID controllers, notch filters, and lead-lag filters~\cite{kundur93a}. This automation system, however, needs to be properly tuned for reliable operation, and adjusted if new components are added or the dynamics of the systems change. 
	Currently, the tuning is triggered and performed manually, typically during the initial setup of new components. The controllers are usually not re-parameterized until large problems in the system make it necessary. 
	Such manual tuning has proven to be sufficient, as the operating power plants had so far not changed significantly over time. For example, oscillatory modes of todays European power system are well known and have quasi-constant frequency and damping~\cite{Grebe10}.
	However, the practice of never touching a running system, as long as it operates sufficiently well, is challenged by renewable power generation.
	
	As the share of renewable generation in large power systems continues to increase, the operation of power systems becomes increasingly challenging~\cite{Ren21Report2018}. For example, the constantly shifting mix of renewable and conventional generation leads to time-varying oscillatory modes~\cite{AlAli14,Preece14}, which are difficult to suppress. If not handled appropriately, todays controllers, which are typically tuned for fixed oscillatory modes, can become ineffective, increasing the risk of blackouts. Thus, new control and optimization methods are necessary in order to improve the robustness of power networks and to account for the changing dynamics. 
	These methods need to be able to cope with increasingly large systems, as power systems can have thousands of components. Additionally, the employed methods need to consider aspects of data security as subsystem operators are typically not willing or cannot fully disclose data about their subsystem with other operators, due to safety, privacy, and economic reasons.
	
	We propose to retune the parameters of the already present controllers in the power systems, to account for the seemingly changing operating conditions. 
	To do so, we present an approach for hierarchical-decentralized structured $\Hinf$ controller synthesis, which is able to optimize the parameters of the existing controllers to the temporally changing operating conditions in the system.
	
	Controller synthesis for power systems often exploits $\Hinf$ or $\Htwo$ optimization, and pole placement, c.f.~\cite{Raoufat16,Pipelzadeh17,Zhu03,befekadu2005robust,MahmoudiNAPS15,Preece13,wu2016input,Schuler14}.
	We note that other control and controller tuning approaches, such as sensitivity analysis~\cite{Marinescu09,Rouco01,Borsche15}, sliding mode controller design~\cite{Liao17}, model predictive control~\cite{Fuchs14}, time-discretization~\cite{lei2001optimization} etc. have been proposed, see, e.g.,~\cite{obaid2017power}, for an overview of different methods for power oscillation damping.
	Very few results, however, consider the optimization of existing controllers~\cite{befekadu2005robust,Marinescu09,kammer2017decentralized}. These approaches require either manual adaptation for each power system~\cite{Marinescu09}, or assume specific dependencies on the controller parameters~\cite{befekadu2005robust,kammer2017decentralized}.
	
	The use of $\Hinf$ optimization methods for controller synthesis has received significant attention in the last decades. 
	We consider the so-called $\Hinf$ structured controller synthesis~\cite{Scherer13,apkarian2018structured}. In structured $\Hinf$ controller synthesis, the controller structure is fixed and only the parameters of the controllers are tuned. Many works are based on the Bounded-real Lemma, e.g.~\cite{Schuler14, Hassibi1999,ibaraki2001rank,dinh2012combining,Han04,Karimi07,schuler2011design,befekadu2006robust,Mesanovic18ACC}.
	Alternative approaches include non-smooth optimization~\cite{HIFOO,apkarian2006nonsmooth} and bisectioning~\cite{kanev2004robust}. For an overview of various optimization methods, see~\cite{sadabadi2016static} and references therein. 
	In recent years, the focus in structured $\Hinf$ optimization shifted towards more efficient methods for finding local minima in structured controller synthesis, as local solutions are often sufficient. These methods often exploit frequency sampling, which makes the synthesis procedure faster~\cite{kammer2017decentralized,apkarian2018structured,boyd2016mimo,mesanovic2018optimalparameter}.
	
	In contrast to previous works, we propose an approach for hierarchical structured $\Hinf$ controller synthesis, tailored towards power systems. 
	Note that, although many works consider the synthesis of distributed controllers, e.g.~\cite{kussner2001damping,schuler2011design}, the synthesis procedure is typically centralized.
	For example, our previous work considered centralized structured controller synthesis~\cite{Mesanovic18ACC, mesanovic2018optimalparameter, Mesanovic17ISGT,mesanovic18ACDC}, whereas here we consider a hierarchical approach. To the authors' knowledge, there are currently no results that consider distributed or hierarchical structured $\Hinf$ synthesis.
	
	The remainder of this work is organized as follows: 
	Section~\ref{sec.PSModel} derives suitable models and formulates the structured $\Hinf$ controller synthesis problem for power systems. Centralized structured $\Hinf$ controller synthesis is briefly reviewed in Section~\ref{sec.CentTuning} before the proposed hierarchical approach is introduced in Section~\ref{sec.DistTuning}. 
	Section~\ref{sec.IEEE68bus} demonstrates the challenge of time-varying eigenmodes on the IEEE 68 bus system and shows the applicability of the proposed approach towards this challenging problem.
	In Section~\ref{sec.ApplicationLarge} the proposed method is applied to a large power system. Finally, conclusions are provided in Section~\ref{sec.Conclusion}.

	\subsection{Mathematical preliminaries}
	\label{subsec.MathPrelim}
	
	We use $\bigSigma(\cdot)$ to denote the largest singular value of a matrix, 
	and $(\cdot)^*$  to denote the conjugate transpose of a matrix. The notation $\succ$ ($\succeq$), and $\prec$ ($\preceq$) is used to denote positive (semi)definiteness and negative (semi)definiteness of  a matrix, respectively. We use $j$ to denote the imaginary unit, $\R_{\ge 0}$ denotes the set of non-negative real numbers, $\C$ denotes the set of complex numbers, and $\C_{>0}$ denotes the set of complex numbers with a positive real part. 
	
	Vectors are denoted with boldface symbols. We use the notation $\myvec{i}(\vx_i)$ to denote the vector obtained by stacking the vectors $\vx_i$ for all $i$, and $\blkdiag{i}(A_i)$ to construct a block-diagonal matrix consisting of matrices $A_i$, for all $i$. The operators $<$ ($\le$), $>$ ($\ge$), and $\left| \cdot \right|$ are defined element-wise for vectors.
	
	\section{Outline and formulation of the controller tuning problem for large power systems}
	\label{sec.PSModel}
	\begin{figure}[tb]
		\centering
		\includegraphics[width=1\columnwidth]{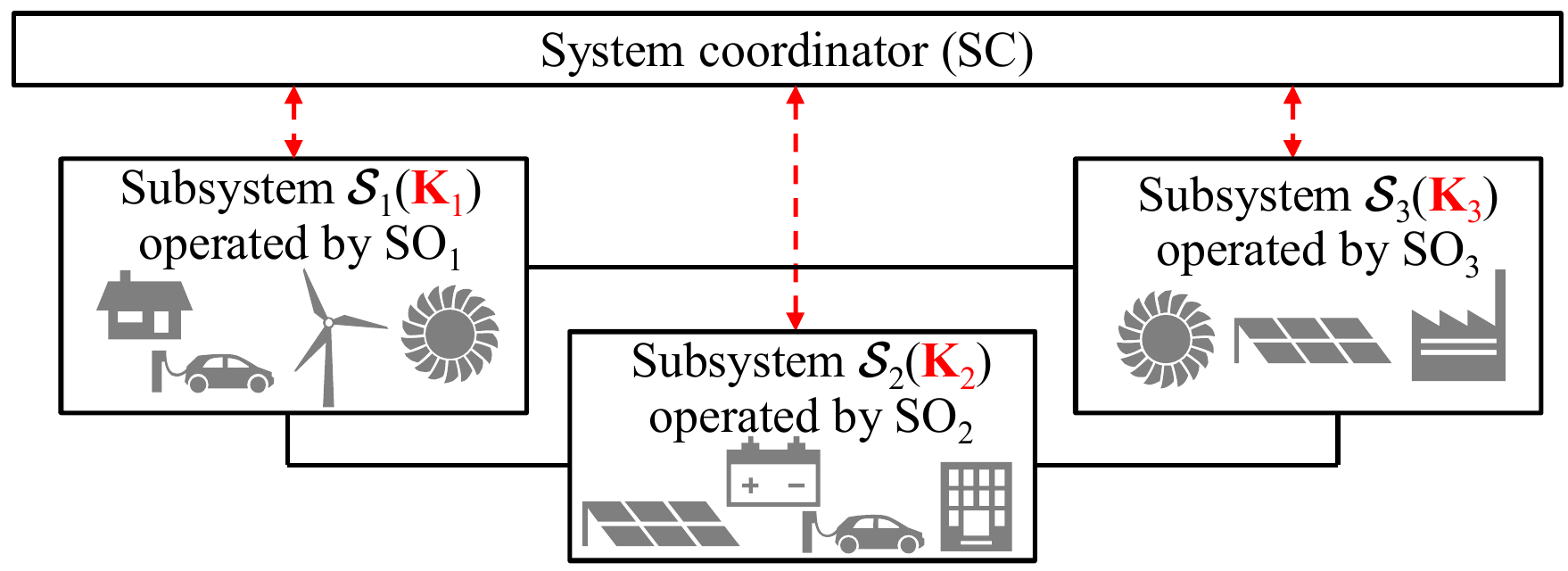}
		\caption{Simple power system example consisting of three interconnected subsystems belonging to different subsystem operators (SOs). Physical connections, i.e. power lines, are denoted with solid lines, whereas (slow) communication links are denoted with dashed lines.}
		\label{fig.GridModel}
	\end{figure}
	
	Figure~\ref{fig.GridModel} shows an exemplary power system, consisting of three subsystems $\cS_1$ - $\cS_3$. 
	The subystems are owned by the subsystems operators (SOs).
	Each subsystem $\cS_i$ contains controllers with tunable parameters, which are collected in parameter vectors $\vK_i$, marked red in Fig.~\ref{fig.GridModel}. The SOs communicate with the system coordinator (SC).
	For example, in Europe, the European Network of Transmission System Operators for Electricity (ENTSOE) could be the system coordinator.
	
	Depending on the infeed of renewable generation and load, the system dynamic behavior changes. In the proposed approach, if the SC or SOs notice that the resiliency of the system decreases, they start the tuning procedure for the controller parameters in the subsystems $\vK_i$. In the proposed hierarchical parameter tuning concept, the SOs only exchange information with the SC, which does not reveal detailed information about each subsystem. The reparameterization process is depicted with red dashed lines in Fig.~\ref{fig.GridModel}. Thereby, only slow communication is needed and limited information exchange. Before explaining the hierarchical tuning procedure in depth in Section~\ref{sec.DistTuning}, we outline the solution approach, as well as the structure, the used dynamic models, and problems of power system control. 
	
	
	\subsection{Outline and basic idea}
	\label{subsec.ApproachOutline}
	\begin{figure}[tb]
		\centering
		\includegraphics[width=0.9\columnwidth]{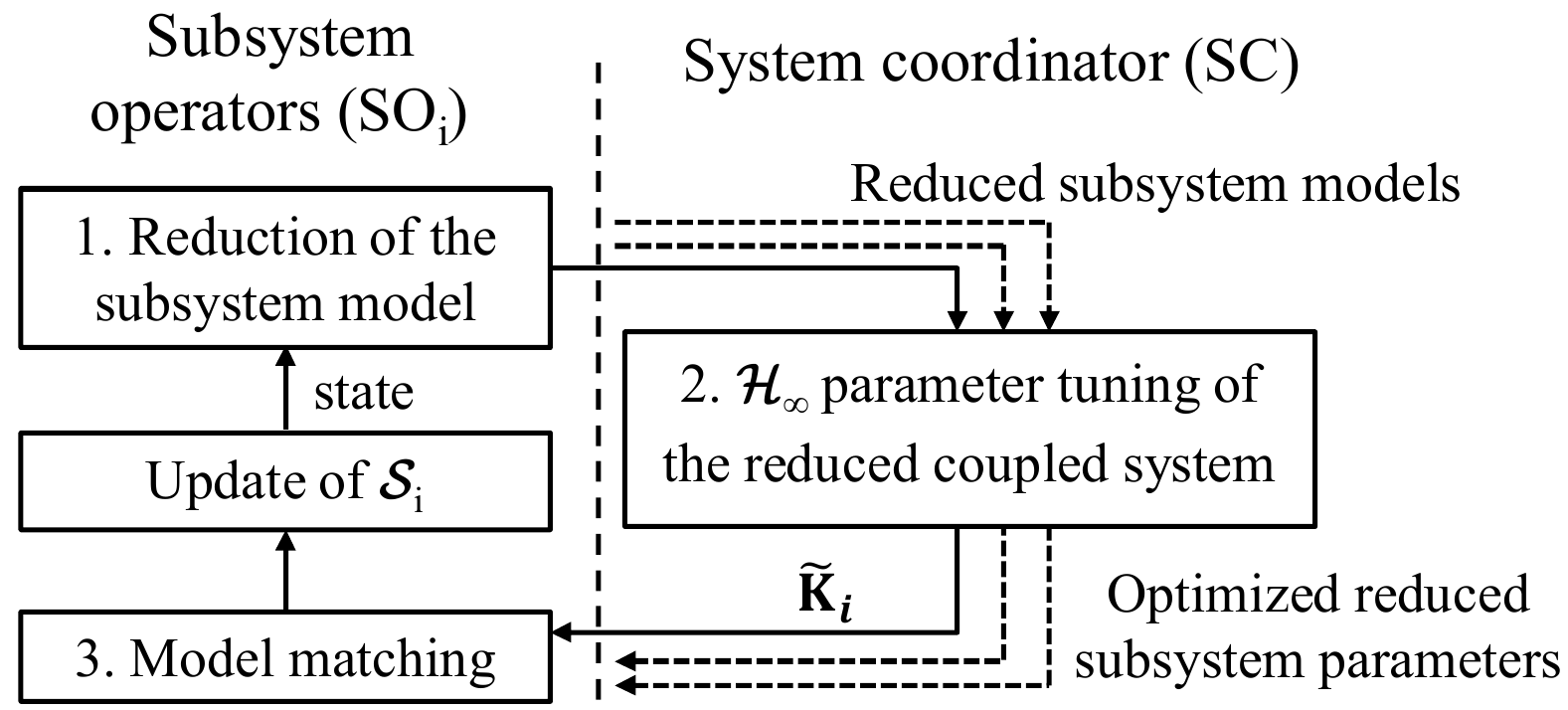}
		\caption{Outline of the proposed hierarchical tuning approach.}
		\label{fig.BasicOutline}
	\end{figure}
	
	In order to improve the readability of subsequent sections, we outline the main idea of the proposed approach, without mathematical details. The complete approach is presented in Section~\ref{sec.DistTuning}.
	
	We consider systems consisting of a number of subsystems $\cS_i$, c.f. Fig.~\ref{fig.GridModel}. Each $\cS_i$ is operated by an SO$_i$, who might not be willing to exchange detailed information about its subsystem with others. The subsystems are coordinated through a system coordinator (SC).
	Figure~\ref{fig.BasicOutline} outlines the proposed iterative controller tuning approach, consisting of three steps. In the first step, each SO reduces its subsystem model, hiding thereby the detailed dynamics and parameters inside the reduced model, leading to increased information security. The reduced model typically has a reduced number of controller parameters for tuning.
	The reduced subsystem models are sent to the SC. The SC combines the reduced subsystem models and tunes the parameters of the resulting reduced overall system. This allows to reduce the computational complexity of the parameter tuning process. The SC sends the optimized reduced parameters of the subsystems back to the SCs. They optimize the parameters of the detailed subsystems to match the reduced models as good as possible, concluding one iteration of the approach. The process is repeated until a satisfactory performance is achieved.
	
	In summary, the increase of data security and reduction of computational complexity is achieved by model reduction in the proposed approach. This, however, leads to a series of challenges which are addressed in more detail in Section~\ref{sec.DistTuning}.
	
	\subsection{Subsystem modeling}
	\label{Subsec.Model}	
	\begin{figure}[tb]
		\centering
		\includegraphics[width=1\columnwidth]{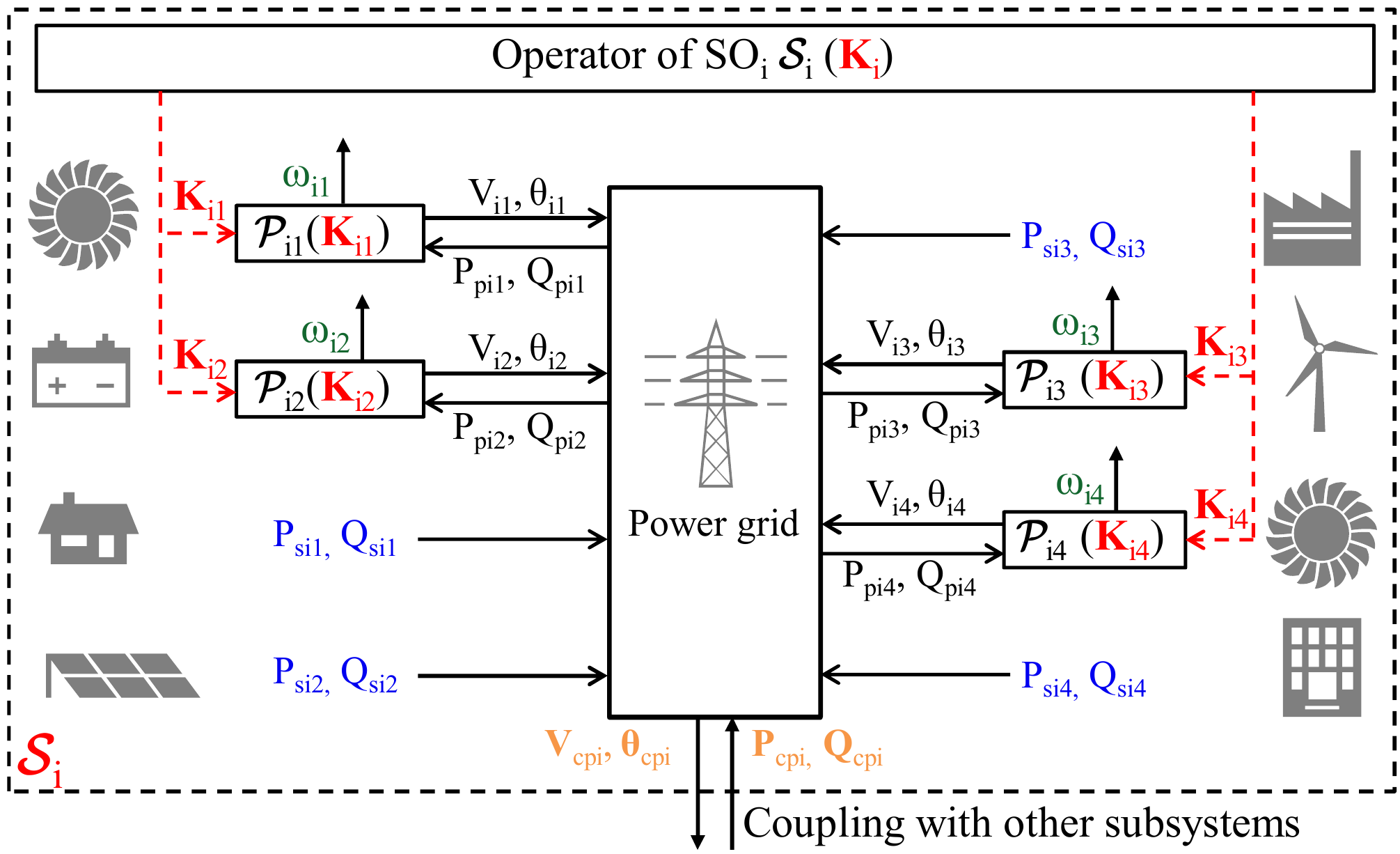}
		\caption{Example subsystem $\cS_i$, operated by the operator SO$_i$. It consists of four dynamic prosumers $\cP_{ij}$ and four static prosumers with power infeeds $P_{sij}$ and $Q_{sij}$. The system operator periodically tunes the controller parameters $\vK_{i} = \myvec{j}(\vK_{ij})$ of the dynamic prosumers, marked red, in order to increase resiliency in the system. The static prosumers, whose infeed is marked with blue, are considered as a disturbance input into the system. The frequencies $\omega_{ij}$ of $\cP_{ij}$, marked green, are used as a part of the performance output of the subsystem. The coupling with other subsystems is realized through the input $\vP_{cpi}$, $\vQ_{cpi}$ and output $\vV_{cpi}$, $\vtheta_{cpi}$, marked orange.}
		\label{fig.SubsystemModel}
	\end{figure}
	
	Figure~\ref{fig.SubsystemModel} shows an exemplary subsystem $\cS_i$. 
	Each subsystem consists of heterogeneous components, such as power plants, renewable generation, storage systems and households. We name these components prosumers, as they can either produce or consume electric power. Thereby, we distinguish between dynamic and static prosumers.
	
	Dynamic prosumers in $\cS_i$, such as power plants, are systems with internal states and dynamics, denoted with $\cP_{ij}$. They posses structured controllers, whose parameters $\vK_{ij}$, marked with red in Fig.~\ref{fig.SubsystemModel}, can be tuned. We consider dynamic prosumers $\cP_{ij}$ which control their voltage magnitude $V_{ij}$ and phase $\theta_{ij}$ at the point of connection, whereas their power infeed into the grid $P_{pij}$ and $Q_{pij}$ are external inputs for the prosumers. This is a standard description, e.g. for conventional power plants with synchronous generators~\cite{kundur93a}, as depicted in Fig.~\ref{fig.SubsystemModel}, where $V_{ij}$ and $\theta_{ij}$ are outputs of $\cP_{ij}$, and $P_{pij}$, $Q_{pij}$ are the external inputs for the prosumers.
	
	Static prosumers have no internal states. They are characterized through their active and reactive power infeed, denoted with $P_{sij}$ and $Q_{sij}$, respectively. Figure~\ref{fig.SubsystemModel} depicts four static prosumers, whose infeeds are marked in blue. We collect the infeeds of static prosumers in $\cS_i$ into vectors $\vP_{si}$ and $\vQ_{si}$, which are considered as external inputs. 
	Static prosumers are used to model two infeed types in power systems. First, they model prosumers with slow dynamics, and whose infeed can be considered constant in the observed time scale. Second, volatile renewable generation and loads are often modeled as static prosumers as well~\cite{poolla2019placement,pddotnuschel2018frequency}. 
	A subset of $\vP_{si}$ and $\vQ_{si}$, representing critical infeeds, models disturbance inputs into the power system model, which we denote with $\vw_{si}$.
	Note that not all elements of $\vP_{si}$ and $\vQ_{si}$ are necessarily a part of $\vw_{si}$. Often a small subset of $\vP_{si}$ and $\vQ_{si}$, representing critical disturbance inputs, is sufficient for the parameter tuning, which reduces computational complexity.
	
	The coupling between the subsystems is represented through the inputs $\vP_{cpi}$ and $\vQ_{cpi}$. They denote the active and reactive powers, respectively, due to the coupling between the $\cS_i$'s. The coupling outputs are the magnitudes $\vV_{cpi}$ and angles $\vtheta_{cpi}$ of the voltage phasors on the border buses (nodes), which are connected through power lines to the other subsystems. We refer to $\vP_{cpi}$ and $\vQ_{cpi}$ as the coupling input $\vw_{ci}$ of $\cS_i$, whereas $\vV_{cpi}$ and $\vtheta_{cpi}$ are the coupling output $\vy_{ci}$. Both are marked orange in Fig.~\ref{fig.SubsystemModel}. Each subsystem can be coupled to other subsystems, i.e. they can have multiple coupling inputs and multiple coupling outputs.
	
	\subsubsection{Power grid model for $\cS_i$}
	\label{subsubsec.PowerGridModel}
	
	The power grid consists of power lines, cables, transformers etc. which interconnect dynamic and static prosumers. In principle, power lines and cables have dynamic states.
	The time constants are, however, orders of magnitude smaller than the relevant dynamics, 
	which are typically slower than 10 Hz~\cite{kundur93a}. For this reason, we neglect the dynamics of the interconnection elements~\cite{kundur93a, schiffer2016survey}. Consequently, the grid, i.e. the power flow, is described by the algebraic power flow equations
	\begin{subequations} 
		\label{eq.PowerFlow}
		\begin{align}
		& P_{ij} = \sum_{k=1}^{N_{Bi}} |V_{aij}||V_{aik}|\big( G_{ci,jk}\cos \Delta\theta_{aijk} + B_{si,jk}\sin \Delta\theta_{aijk} \big)\label{eq.activepower} \\
		& Q_{ij} = \sum_{k=1}^{N_{Bi}}  |V_{aij}||V_{aik}| \big( G_{ci,jk}\sin \Delta\theta_{aijk} - B_{si,jk}\cos \Delta\theta_{aijk}  \big), \label{eq.reactivepower}
		\end{align}
	\end{subequations}
	where $j = 1...N_{Bi}$, $N_{Bi}$ is the number of buses (nodes) in the grid of $\cS_i$,
	$P_{ij} = P_{pij} + P_{sij}+P_{cij}$ and $Q_{ij} = Q_{pij} + Q_{sij}+Q_{cij}$ are the injected active and reactive powers into the j-th bus (node) of the grid in $\cS_i$ by a dynamic prosumer ($P_{pij}$, $Q_{pij}$), static prosumer ($P_{sij}$, $Q_{sij}$) or by a coupling branch ($P_{cij}$, $Q_{cij}$),
	$V_{aij}$ and $\theta_{aij}$ are the magnitude and angle of the voltage phasor at the j-th bus of $\cS_i$, including buses of dynamic prosumers $V_{ij}$, $\theta_{ij}$ and coupling buses $V_{cpij}$, $\theta_{cpij}$,
	$\Delta\theta_{aijk} = \theta_{aij} - \theta_{aik}$,
	and $G_{ci,jk}$ and $B_{si,jl}$ are the elements of the  conductance and susceptance matrix of the power grid in $\cS_i$~\cite{kundur93a}. 
	
	\subsubsection{Power plant modeling}
	\label{subsubsec.ProsumerModel}
	\begin{figure}[tb]
		\centering
		\includegraphics[width=0.7\columnwidth]{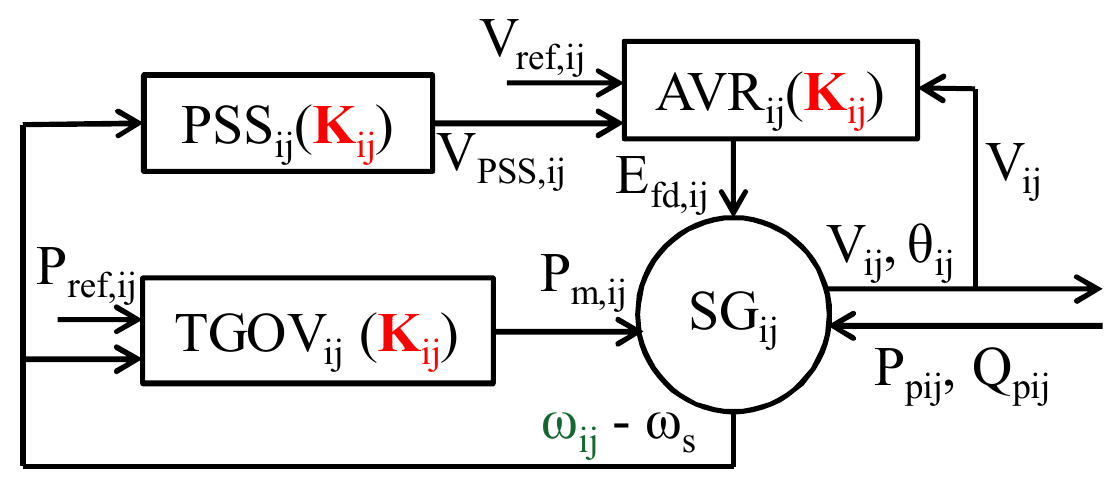}
		\caption{Model of a dynamic prosumer $\cP_{ij}$, a power plant. It consists of a synchronous generator (SG$_{ij}$), automatic voltage regulator and exciter (AVR$_{ij}$), power system stabilizer (PSS$_{ij}$), and of a turbine and governor model (TGOV$_{ij}$).}
		\label{fig.SGModel}
	\end{figure}
	We allow to include arbitrary models of dynamic prosumers. We outline the model of a power plant as an example of a dynamic prosumer, as it is used in the subsequent simulation studies. We consider power plants that consist of a synchronous generator with controllers and actuators, as shown in Fig.~\ref{fig.SGModel}. 
	The dynamics is captured by a 6-th order model for the synchronous generator (SG$_{ij}$)~\cite{kundur93a}.
	
	The automatic voltage regulator and exciter (AVR$_{ij}$) control the voltage at the power plant terminals $V_{ij}$ to be as close as possible to a reference value $V_{ref,ij}$. The output of AVR$_{ij}$ is the field winding voltage $E_{fd,ij}$ which is an input of the SG$_{ij}$.
	As automatic voltage regulators can reduce the stability margin in power systems~\cite{kundur93a}, power plants can be equipped with power system stabilizers (PSS$_{ij}$). PSSs are analogue or digital controllers, with the task to improve the system stability and increase the damping of oscillations in power systems.
	We consider that the PSS$_{ij}$ takes as input the deviation of the generator frequency $\omega_{ij}$ from the nominal system frequency $\omega_s$, while its output $V_{PSS,ij}$ is an additional input of the AVR$_{ij}$.
	The governor and turbine (TGOV$_{ij}$) control the generator frequency by adapting the mechanical power $P_{m,ij}$ transfered to the synchronous generator.
	
	In practice, many different controllers are used, see e.g.~\cite{IEEEExciters06}. All controllers (AVR$_{ij}$, PSS$_{ij}$ and TGOV$_{ij}$), however, contain tunable controller parameters. Examples for these controllers are provided in ~\ref{App.IEEE68Models}.
	
	\subsubsection{Performance outputs for the controller tuning}
	\label{subsubsec.PerfOutput}
	
	Performance outputs are used in the controller tuning procedure to define an objectively quantifiable tuning goal. In power systems, the frequencies of dynamic prosumers $\omega_{ij}$ are typically used to asses the system performance~\cite{kundur93a}. They are defined by $\omega_{ij} = \dot{\theta}_{ij}$, where $\theta_{ij}$ is the angle of the voltage phasor of $\cP_{ij}$, marked green in Figs.~\ref{fig.SubsystemModel} and~\ref{fig.SGModel}.
	Consequently, we choose the vector of all frequencies (or a subset) as a possible performance output~\cite{mesanovic2018optimalparameter}
	\begin{align}
	\vomega_i = \begin{pmatrix} \omega_{i1} & ... & \omega_{iN_{Di}}\end{pmatrix} ^T, \label{eq.omegaAllOut}
	\end{align}
	where $N_{Di}$ is the number of dynamic prosumers in $\cS_i$.
	When oscillations between subsystems are considered, prosumers in one subsystem oscillate against prosumers of other subsystems. Thus, choosing the so-called center-of-inertia frequency~\cite{UlbigLowInertiaImpact} as the performance output for one subsystem is a standard choice when oscillations between subsystems should be suppressed
	\begin{align}
	\omega_{COIi} = \sum_{j = 1}^{N_{Di}} J_{ij} \omega_{ij} \Big/  \sum_{j = 1}^{N_{Di}} J_{ij}. \label{eq.COIFreq}
	\end{align}
	Here $J_{ij}$ is the inertia of the $j$-th dynamic prosumer. The center-of-inertia frequency frequency represents the weighted arithmetical mean of all generator frequencies.
	
	\subsubsection{Overall model of subsystem $\cS_i$}
	\label{subsubsec.SubsysModel}
	
	Combining the power grid equations~\eqref{eq.PowerFlow} with the prosumer models, the dynamical model of $\cS_i$ is given by a differential-algebraic nonlinear model
	\begin{subequations} \label{eq.nonlinearModel}
		\begin{align}
		\dot{\vx}_i =& f_i(\vx_i, \vw_{si}, \vw_{ci},\vK_i) \\
		0 =& h_i(\vx_i, \vw_{si},\vw_{ci},\vK_i). \label{eq.nonlinearModel.alg}
		\end{align}
	\end{subequations}
	Here $\vx_i \in \R^{\cdot N_{xi}}$ combines all dynamic prosumer states of $\cS_i$,
	$\vw_{si}$ is the vector of inputs from static prosumers, represented by a subset of $\vP_{si}$ and $\vQ_{si}$, 
	$\vw_{ci} \in \R^{n_{ci}}$ is the vector of coupling inputs, represented by $\vP_{cpi}$ and $\vQ_{cpi}$,
	$\vK_{i} \in \R^{N_{ki}}$ is the vector of tunable controller parameters of all dynamic prosumers in $\cS_i$,
	$f_i$ describes the prosumer dynamics in $\cS_i$,
	and $h_i$ represents the power flow equation~\eqref{eq.PowerFlow} in $\cS_i$. All considered outputs of~\eqref{eq.nonlinearModel} are linear combinations of the elements in $\vx_i$.
	
	We focus on the small-signal behavior of the system around a nominal operating point. Thus, we linearize~\eqref{eq.nonlinearModel} around the known steady-state $\vx_{i0}$ with the known inputs $\vw_{si0}$ and $\vw_{ci0}$. While this is a first-order approximation for small-signal deviations, it allows us to use methods developed for linear systems. It is furthermore often sufficient even in case of large-scale disturbances~\cite{mesanovic2018optimalparameter,mesanovic2018ISGT,poolla2019placement}. 
	\begin{rem}[steady-states]
		The steady state of the power grid in $\cS_i$ can be estimated based on measurements of voltages and power infeeds~\cite{gomez2004power}, for which commercial tools exist. This allows to obtain initial values for the steady-states of dynamic prosumers.
	\end{rem}
	The linearized algebraic equations obtained from~\eqref{eq.nonlinearModel.alg} have full rank, leading to 
	\begin{align} \label{eq.linearizedModel1}
	\dot{ {\Delta \vx}}_i &= \widetilde A_i(\vK_i) {\Delta \vx_i} +  \widetilde B_{ci} (\vK_i) {\Delta \vw_{ci}}  +  \widetilde B_{si} (\vK_i) {\Delta \vw_{si}},
	\end{align}
	where $\Delta \vx_i = \vx_i - \vx_{i0}$, $\Delta \vw_{ci} = \vw_{ci} - \vw_{ci0}$, and  $\Delta \vw_{si} = \vw_{si} - \vw_{si0}$. Subsequently, we skip $\Delta$ for notational convenience, i.e. we use deviations from the linearization point in the rest of this work.
	We note that the system matrix $\widetilde A_i(\vK_i)$ has an eigenvalue at zero, as the coupling power flow equation~\eqref{eq.PowerFlow} is invariant under offsets $\widetilde \theta_{ij} = \theta_{ij} + \delta \theta$, where $\delta  \theta \in \R$ is identical for all $j$. Reducing this zero eigenmode~\cite{wu2016input}, leads to
	\begin{align} \label{eq.linearizedModel}
	\dot{ \vx}_i &=  A_i(\vK_i)  \vx_i +   B_{ci} (\vK_i)  \vw_{ci}  +   B_{si} (\vK_i)  \vw_{si}.
	\end{align}
	We note that~\eqref{eq.linearizedModel} is exponentially stable for a suitable choice of $\vK_i$. This assumption does not restrict applicability, as all practical power systems have controllers which stabilize the system.
	
	The proposed approach is based on different transfer functions, for which we need to define the outputs of interest, c.f. Fig.~\ref{fig.BasicOutline}.
	First, a transfer function is needed for each subsystem which couples the disturbance and coupling inputs $\vw_{ci}$, $\vw_{si}$ and outputs $\vy_{ci} = \begin{pmatrix} \vV_{cpi}^T & \vtheta_{cpi}^T  \end{pmatrix}^T$, $\vy_{pi} = \omega_{COIi}$, defined as
	\begin{align}
	\begin{pmatrix}
	\vy_{ci} \\ \vy_{pi}
	\end{pmatrix} &= G_i(\vK_i, s) \begin{pmatrix}
	\vw_{ci} \\ \vw_{si}
	\end{pmatrix} \nonumber \\&= \begin{pmatrix}
	G_{cci}(\vK_i, s) & G_{csi}(\vK_i, s) \\ G_{pci}(\vK_i, s) & G_{psi}(\vK_i, s)
	\end{pmatrix} \begin{pmatrix}
	\vw_{ci} \\ \vw_{si}
	\end{pmatrix}. \label{eq.TFSubsystem}
	\end{align}
	The transfer function $G_{csi}(\vK_i, s)$ represents algebraic power flow equations which couple the power of static prosumers $\vw_{si}$ with the coupling output $\vy_{ci}$. Thus, $G_{csi}(\vK_i, s)$ is not a function of controller parameters and is written as a constant matrix, i.e. $G_{csi}(\vK_i, s) = M_{si}$.
	
	The second transfer function, $G_i^*(\vK_i,s)$, is used for the $\Hinf$ controller tuning, see Fig.~\ref{fig.BasicOutline}.
	For this purpose, the vector of all dynamic prosumer frequencies is defined as the performance output, described with~\eqref{eq.omegaAllOut}. We consider $ \vw_{ci}$ and $ \vw_{si}$ as disturbance inputs, because disturbances can come from other subsystems through the coupling input, as well as from internal static prosumers. Thereby, if $\cS_i$ is not coupled with other systems, i.e. if $\cS_i$ represents an isolated system, the vector $\vw_{ci}$ is empty.
	Thus, $G_i^*(\vK_i,s)$ is in this case given by
	\begin{align}
	\vy_{pi}^* &= \vomega_{i} = G_i^*(\vK_i,s) \begin{pmatrix}  \vw_{ci} \\  \vw_{si} \end{pmatrix}, \label{eq.TFOmegas}
	\end{align}

	Further details on the use of the transfer functions can be found in Section~\ref{sec.DistTuning}.
	
	\subsection{Coupled power system model}
	\label{Subsec.CoupledModel}	
	\begin{figure}[tb]
		\centering
		\includegraphics[width=0.55\columnwidth]{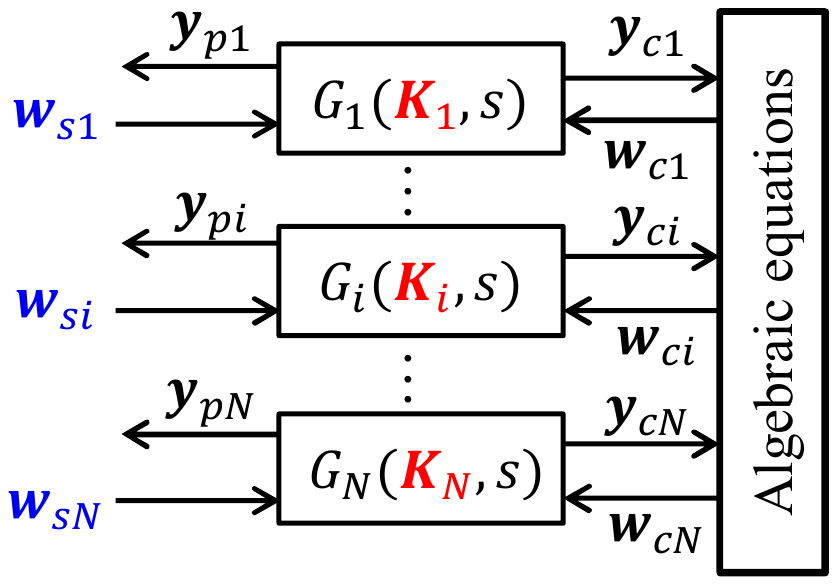}
		\caption{Coupled power system model consisting of N subsystems $\cS_i$.}
		\label{fig.SysStructOutline}
	\end{figure}
	
	The previous subsection outlined the modeling and structure for a single $\cS_i$. We now outline how the subsystems are coupled to obtain the model of the entire power system.
	Physically, the subsystems are coupled via power lines, which we model using linearized algebraic power flow equations~\eqref{eq.PowerFlow}, c.f. Fig.~\ref{fig.SysStructOutline}.
	The coupling between the subsystems is given by
	\begin{align} \label{eq.LinCoupling}
	\vw_c = M \vy_c,
	\end{align}
	where $ \vw_c = \myvec{i} (\vw_{ci})$, $\vy_c = \myvec{i} (\vy_{ci})$, and M represents the linearized coupling between the subsystems and is a full-rank matrix.
	The transfer function of the coupled system, i.e. $G(\vK,s)$, from the disturbance input $\vw_s = \myvec{i} (\vw_{si})$ to the performance output $\vy_p = \myvec{i}(\vy_{pi})$, which quantifies oscillations between subsystems, is obtained by combining ~\eqref{eq.LinCoupling} with $G_{ci}$ and $G_{pi}$ in~\eqref{eq.TFSubsystem}
	\begin{align}
	&G(\vK,s) = \blkdiag{i} (G_{psi}(\vK_i, s)) \\ 
	&+ \blkdiag{i} (G_{pci}(\vK_i, s)) \left( I - M \blkdiag{i} (G_{cci}(\vK_i, s)) \right) ^{-1} M \blkdiag{i} (M_{si}) \nonumber
	\label{eq.DetailedCoupledSys}
	\end{align}
	where $\vK = \myvec{i}(\vK_i)$. The transfer function $G(\vK, s)$ has several challenging properties: (1) The dependency of G on $\vK$ is nonlinear, (2) the coupled system can be very large with thousands of states, (3) no single entity knows the parameters and detailed structure of $G(\vK, s)$. Only the SO of $\cS_i$ is aware of all parameters of $G_i(\vK_i, s)$. The SC knows only the parameters of the coupling matrix $M$.
	
	Given this setup, we can formulate the main research question considered in this work: what are the optimal parameters $\vK$ to minimize the $\Hinf$ norm of $G(\vK,s)$?
	
	
	We will first review a centralized approach before introducing the proposed hierarchical approach.
	
	\section{Centralized  controller tuning}
	\label{sec.CentTuning}
	
	One way to tune $\vK$ is by centralized tuning. In the following, we outline the centralized $\Hinf$ tuning algorithm presented in~\cite{mesanovic2018optimalparameter}, because it serves as the basis for the proposed hierarchical tuning.
	The $\Hinf$ norm of a stable system $G(s)$, denoted with $\|G(s)\|_\infty$~\cite{boyd1985subharmonic} is defined by
	\begin{align}
	\!\!\!\!\! \|G(s)\|_\infty \! \DefinedAs \!  \text{sup}_{s \in \C_{>0}} \: \bigSigma \: (G(s)) = \text{sup}_{\omega \in \R_{\ge 0}} \: \bigSigma \: (G(j \omega)).
	\end{align}
	We propose to use the $\Hinf$ norm as the optimization criterion, because it represents the maximal amplification of amplitude of any harmonic input signal in any output direction.
	Thus, minimizing the $\Hinf$ norm minimizes the worst-case amplification of oscillation frequencies. 
	Minimizing the $\Hinf$ norm thus allows to improve robustness of systems.
	We minimize $\|G(\vK, s)\|_\infty$ by optimizing the vector of parameters $\vK$.
	
	The centralized tuning approach is based on the following theorem:
	\begin{thm}[Centralized tuning] \cite{mesanovic2018optimalparameter} \label{thm.CentTuningStab}
		Given a detectable multiple-input-multiple-output system which is a continuous nonlinear function of the vector of tunable controller parameters $\vK$, denoted by $G(\vK, s)$.
		Furthermore, given an initial, exponentially stabilizing parameterization $\vK_{0}$. Assuming there are no cancellations of parameter-dependent poles and zeros on the imaginary axis, there exists a sufficiently large discrete set of frequencies $\Omega$, such that the solution of
		\begin{subequations} \label{eq.FiniteOptProb}
			\begin{align}
			\!\!\!\!\!\!\!\min_{\gamma, \vK}  & \quad \gamma \\
			\!\!\!\!\!\!\!\text{s.t.} \quad & \begin{pmatrix}
			\gamma I     & G(\vK, j\omega_k) \\
			G(\vK, j\omega_k)^* & \gamma I
			\end{pmatrix} \succ 0, \quad \forall \omega_k \in \Omega \label{eq.finiteConstraints} \\
			& \ul{\vK} \leq \vK \leq \ol{\vK}
			\end{align}
		\end{subequations}
		is a stabilizing controller which minimizes $\HinfNorm{G(\vK, s)}$ for the set of frequencies $\Omega$. Here $\ul{\vK}$ and $\ol{\vK}$ are box constraints on the controller parameters, which may be $\pm \infty$.
	\end{thm}
	We refer to~\cite{mesanovic2018optimalparameter} for a proof of Theorem~\ref{thm.CentTuningStab}.
	
	Even though centralized tuning improves the performance of the system and eliminates oscillations, it is difficult to be applied in practice for large systems: due to safety and privacy reasons, there will be no entity which has access to all parameters for large, interconnected power systems like the European power system or the western interconnection in the US, i.e. of $G(\vK,s)$. The centralized approach is furthermore limited with respect to the maximal size of the system which can be optimized. Hierarchical structured $\Hinf$ controller tuning, introduced in the next section, simultaneously addresses both issues: it provides better scalability of the approach while simultaneously increasing data privacy.
	
	
	\section{Hierarchical parameter tuning}
	\label{sec.DistTuning}
	
	As outlined previously, centralized tuning is often undesired or not possible. Thus, we introduce in this Section the hierarchical decentralized tuning approach. The hierarchical tuning is based on the ideas of Fig.~\ref{fig.BasicOutline}.
	To allow scalability and to increase privacy, we propose the following steps in the k-th iteration of the algorithm, compare also Fig.~\ref{fig.OptOutline}:
	\begin{enumerate}
		\item Model reduction of the subsystems: all SOs first calculate a reduced model of their subsystem, denoted with $\tG_i^\kit(\vtK_{i,init}^\kit,s)$, based on the detailed model $G_i(\vK_{i,init}^\kit,s)$, which is defined in~\eqref{eq.TFSubsystem}. Here $\tG_i$ and $\vtK_i$ denote the reduced model of $\cS_i$ and its parameter vector, respectively, and $(k)$ denotes the iteration. The detailed parameters and dynamics of the individual prosumers and of the power grid are "hidden" in $\tG_i$ and $\vtK_i$, which leads to increased data privacy. The SOs send $\tG_i^\kit$ and $\vtK_{i,init}^\kit$ to the SC.
		
		\item Centralized $\Hinf$ parameter tuning based on the reduced subsystem models: the system coordinator (SC) couples the reduced subsystems and calculates a reduced model of the whole system $\tG^\kit(\vtK,s)$. It optimizes the reduced system, which has a smaller complexity than an overall detailed model capturing the detailed dynamics of all subsystems. The optimized reduced set of parameters $\vtK_{i,opt}^\kit$ are sent back to the SOs. \label{step.Optimization}
		
		\item Model matching: The optimized reduced subsystems $\tG_i^\kit(\vtK_{i,opt}^\kit,s)$ serve as a reference model for the SOs. They optimize the parameters of their respective detailed models $G_i(\vK_i,s)$ to match the reference model to the best possible extent.
	\end{enumerate}
	This iterative process is performed repeatedly until a stopping criterion is fulfilled. Step~\ref{step.Optimization}, i.e. the optimization step based on the reduced models, contains the main idea of the approach: instead of optimizing the entire detailed system, the SC uses the reduced models, leading to improved data privacy and scalability of the approach. 
	However, the optimization of reduced models introduces challenges with respect to the model reduction and model matching steps. In general, $\vtK_{init}^\kit$ is not the same as $\vtK_{i,opt}^\kitprev$, as $\vtK_{init}^\kit$ is obtained from the model reduction step. 
	
	In the next sections, we provide details for each step of the approachs. We start with the optimization step in order to better clarify the main idea of the approach. Afterwards, we detail the model matching step before turning to the model reduction step, which is detailed last as the requirements are derived from the optimization and model matching steps.
	
	\begin{figure}[tb]
		\centering
		\includegraphics[width=1\columnwidth]{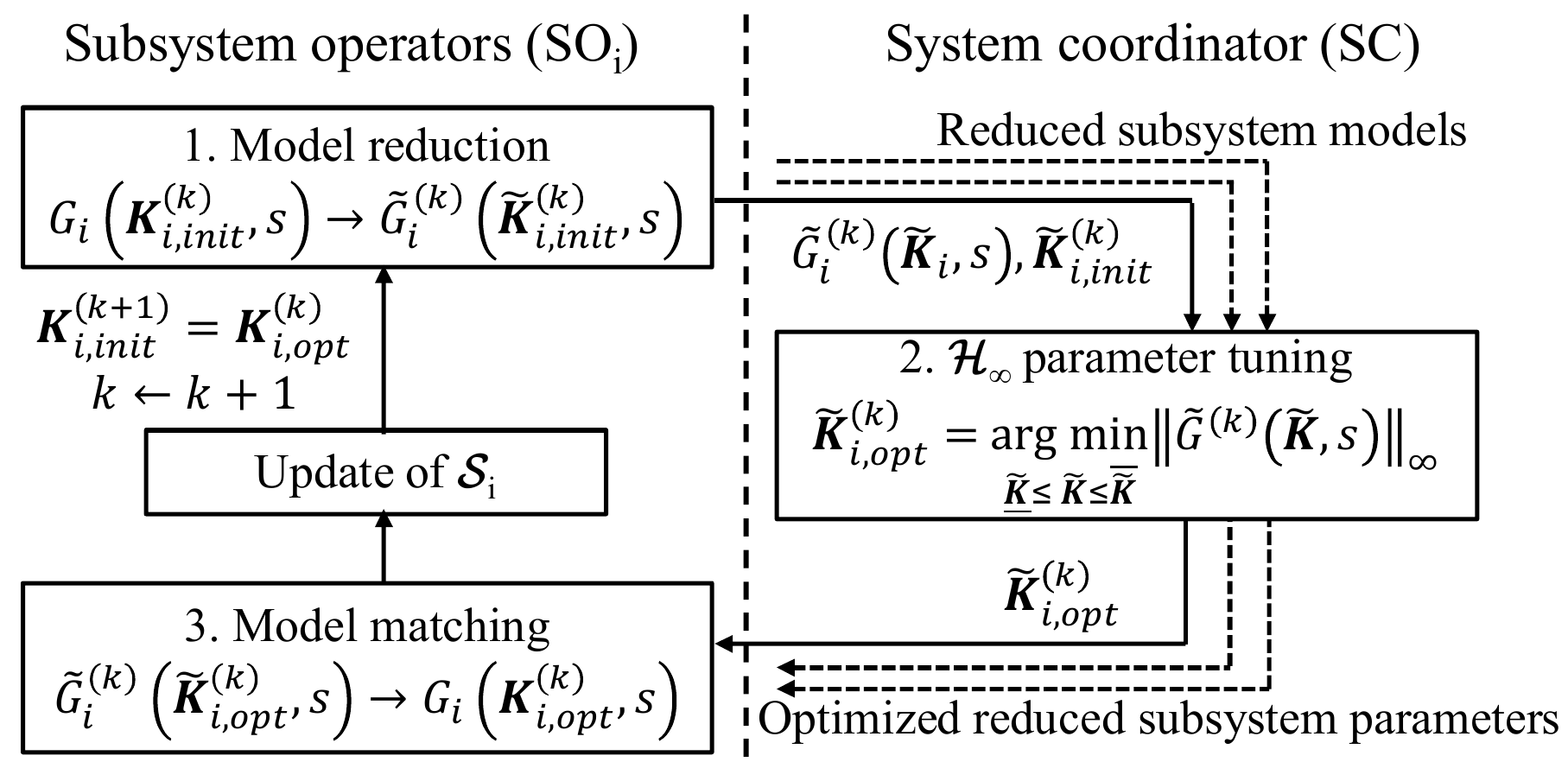}
		\caption{Details of the proposed optimization approach.}
		\label{fig.OptOutline}
	\end{figure}
	
	\subsection{Structured $\Hinf$ parameter tuning of the reduced system}
	
	For the tuning, the SC receives reduced models of each subsystem, i.e. $\tG_i^\kit(\vtK_i,s)$, 
	together with initial reduced parameter vectors $\vtK_{i,init}^\kit$ from each SO. As the SC knows the coupling matrix $M$, the overall reduced model can be formulated
	\begin{align}\label{eq.CoupledRedSys}
	&\tG(\vtK,s) = \blkdiag{i} (\tG_{psi}(\vtK_i, s)) \\ 
	&+ \blkdiag{i} (\tG_{pci}(\vtK_i, s)) \left( I - M \blkdiag{i} (\tG_{cci}(\vtK_i, s)) \right) ^{-1} M \blkdiag{i} (M_{si}) \nonumber
	\end{align}
	Based on this model, the SC performs the structured $\Hinf$ controller tuning
	\begin{align}
	\vtK_{opt}^\kit = \myvec{i} (\vtK_{i,opt}^\kit) = \argmin_{ \ul{\vtK} \le \vtK \le \ol{\vtK} } \left\| \tG(\vtK, \jw) \right\|_\infty \label{eq.StructReducedOpt}
	\end{align}
	with~\eqref{eq.FiniteOptProb}, 
	where $\ul{\vtK}$ and $\ol{\vtK}$ are box constraints on the reduced vector of controller parameters. The optimization is initialized with $\vtK_{init}^\kit = \myvec{i}(\vtK_{i,init}^\kit)$. The optimized parameter vectors $\vtK_{i,opt}^\kit$ are sent to the SOs.
	
	\subsection{Model matching step}
	
	When the SO of the i-th subsystem receives the reduced optimized parameter vector $\vtK_{i,opt}^\kit$ from the SC, it uses $\tG^\kit(\vtK_{i,opt}^\kit,\jw)$ as the reference model to adapt the parameters of the full subsystem, such that the detailed model best matches the optimized reduced model. The SO$_i$ solves an $\Hinf$ model matching problem to obtain the vector of detailed controller parameters of $\cS_i$ in the k-th iteration
	\begin{align}
	\!\!\!\!\!\! \vK_{i,opt}^\kit \!=\! \argmin_{\ul{\vK}_{i} \le \vK_{i} \le \ol{\vK}_{i}} &\left\| G_i(\vK_i, \jw) - \tG_i^\kit(\vtK_{i,opt}^\kit, \jw) \right\|_\infty \label{eq.ModelMatching}
	\end{align}
	where $\ul{\vK}_{i}$ and $\ol{\vK}_{i}$ represent box constraints on the controller parameters. From the solution of this problem, the optimized detailed parameter vector in the k-th iteration $\vK_{i,opt}^\kit$ is obtained for each area. This parameter vector is used as the initial value for the model reduction for the next updated iteration.
	
	\subsection{Model reduction step}
	From the previous two steps, the following requirements for the reduced model $\tG_i^\kit(\vtK,\jw)$ are evident:
	\begin{enumerate}
		\item The error between $\tG_i^\kit(\vtK^\kit_{i,init},\jw)$ and $G_i(\vK_{i,init}^\kit,\jw)$ should be small.
		\item The reduced models need to have tunable parameters $\vtK_i^\kit$ which can be optimized in the optimization step.
		\item The $\Hinf$ error between $\tG_i^\kit(\vtK_{i,opt}^\kit,\jw)$ and $G_i(\vK_{i,opt}^\kit,\jw)$ after the model matching step should be small. This means that the reduced model needs to have a representative, realistic dependency on the vector of reduced tunable parameters $\vtK_i$, i.e. similar to the dependency of the detailed model on the full parameter vector $\vK_i$.
	\end{enumerate}
	Due to the second requirement, the application of unstructured model reduction approaches, such as balanced model order reduction, is difficult, as they do not retain parametric dependencies in the reduced model. This can be overcome, e.g., by introducing an additional static state-feedback controller for the reduced model whose parameters can be optimized. However, such a model may have very different dynamic properties after the optimization step, meaning that the model matching step may not allow to reduce the error between the detailed model and the reduced model, i.e. the third requirement may not be satisfied.
	
	Instead, we propose to perform the model reduction step by selecting a structured reference model for one area. The reference model has a fixed structure and is parameterized to match the detailed model. Structured reduced models have been widely used in the power system community, replacing groups of interconnected power plants by a small number of power plants, see e.g. in~\cite{chow1995inertial}. Every SO obtains its reduced model $\tG_i$ by solving the following parameter matching problem
	\begin{align}
	&(\vR_i^\kit\!\!,\vtK_{i,init}^\kit)\!\!
	%
	%
	=\!\!\!\!\!\argmin_{\substack{ \ul{\vR}_i \le \vR_i \le \ol{\vR}_i \\ \ul{\vtK}_i \le \vtK_i \le \ol{\vtK}_i  }} \!\!\left\|  G_{i}'(\vR_i,\vtK_i,\jw)   \!\!-\!\!  G_{i}(\vK_{i,init}^\kit,\jw) \right\|_\infty\!\!\!\!.
	\label{eq.ModelReductionProblem}
	\end{align}
	Here $G_{i}'$ is the reduced model with the predefined structure, and $\vR_i$ is the vector of other model parameters.
	For power systems, $\vR_i$ can represent, for example, the physical parameters of TGOV$_{ij}$, AVR$_{ij}$, and PSS$_{ij}$.
	The reduced model becomes
	\begin{align}
	\tG_{i}^\kit(\vtK_{i}, \jw) \DefinedAs G_{i}' (\vR_i^\kit, \vtK_i, \jw).
	\end{align}
	The model reduction is outlined in detail in Section~\ref{subsec.StructRed4PowSys}.
	
	\subsection{Condition for the improvement of the system $\Hinf$ norm}
	\label{subsec.ImprovCond}
	
	The minimization of $\HinfNorm{\tG(\vtK,j\omega)}$ does not guarantee that $\HinfNorm{G(\vK,j\omega)}$ will be reduced per iteration. 
	For this reason, we introduce the following Lemma:
	\begin{lem}
		$\HinfNorm{G(\vK_{init}^\kit,s)}$ is reduced if and only if
		\begin{align}
		& \max_{\omega \in \R} \: \: \epsilon_{opt}^\kit(\omega) - \epsilon_{init}^\kit(\omega) + \alpha^\kit (\omega) + \bigSigma(G(\vK_{init}^\kit, \jw)) \nonumber \\ 
		& \qquad < \max_{\omega \in \R} \: \:  \bigSigma(G(\vK_{init}^\kit, \jw)). \label{eq.ImprovCond}
		\end{align}
		Here $\epsilon_{init}^\kit(\omega)$ denotes the singular value error between the detailed model $G(\vK_{init}^\kit, \jw)$ and the reduced model $\tG(\vtK_{init}^\kit, \jw)$ after the model reduction step
		\begin{align}
		\epsilon_{init}^\kit(\omega) = \bigSigma(G(\vK_{init}^\kit, \jw)) - \bigSigma( \tG(\vtK_{init}^\kit, \jw)). \label{eq.RedErr}
		\end{align}
		Furthermore, $\epsilon_{opt}^\kit(\omega)$ is the error term which occurs after the model matching step between $\tG(\vtK_{opt}^\kit, \jw)$ and $G(\vK_{opt}^\kit, \jw)$
		\begin{align}
		\epsilon_{opt}^\kit(\omega) = \bigSigma(G(\vK_{opt}^\kit, \jw)) - \bigSigma( \tG(\vtK_{opt}^\kit, \jw)), \label{eq.MatchErr}
		\end{align}
		and $\alpha^\kit(\omega)$ quantifies the change in the singular values of the reduced system during the parameter tuning in the k-th iteration
		\begin{align}
		\alpha^\kit (\omega) = \bigSigma(\tG(\vtK_{opt}^\kit, \jw)) - \bigSigma(\tG(\vtK_{init}^\kit,\jw)). \label{eq.AlphaOpt}
		\end{align}
	\end{lem}
	\begin{proof}
		Reduction of $\HinfNorm{G(\vK_{init}^\kit,s)}$ in an iteration is equivalent to the condition
		\begin{align}
		\| G(\vK_{opt}^\kit,\jw) \|_\infty &< \| G(\vK_{init}^\kit,\jw) \|_\infty \\
		\Leftrightarrow  \max_{\omega \in \R} \: \: \bigSigma(G(\vK_{opt}^\kit,\jw)) &< \max_{\omega \in \R} \: \: \bigSigma (G(\vK_{init}^\kit,\jw)). \label{eq.BasicCond4Succ}
		\end{align}
		Combining Equations~\eqref{eq.RedErr},~\eqref{eq.MatchErr}, and~\eqref{eq.AlphaOpt}, leads us to the following relation
		\begin{align}
		\label{eq.SigmaOpt2SigmaInit}
		\bigSigma(G(\vK_{opt}^\kit, \jw))  =&  \epsilon_{opt}^\kit(\omega) + \alpha^\kit (\omega) - \epsilon_{init}^\kit(\omega) \nonumber \\
		& +\bigSigma(G(\vK_{init}^\kit, \jw)).
		\end{align} 
		Inserting~\eqref{eq.SigmaOpt2SigmaInit} into~\eqref{eq.BasicCond4Succ}, we obtain~\eqref{eq.ImprovCond}. \end{proof}
Note that $\alpha^\kit (\omega)$ does not need to be $<$0 for all $\omega \in \R$. Rather, the following relation must be satisfied for successful $\Hinf$ norm minimization
\begin{align}
\max_{\omega \in \R} \: \: \alpha^\kit (\omega) + \bigSigma(\tG(\vtK_{init}^\kit,\jw)) < \max_{\omega \in \R} \: \:  \bigSigma(\tG(\vtK_{init}^\kit,\jw)). \label{eq.ReducedImprovCond}
\end{align}
Ideally, $\epsilon_{init}^\kit(\omega)$ and $\epsilon_{opt}^\kit(\omega)$ would be equal to 0 for all frequencies, i.e. no error is introduced during model reduction and model matching. This reduces~\eqref{eq.ImprovCond} to~\eqref{eq.ReducedImprovCond}, making optimization of the reduced and detailed models equivalent. However, in general, $\epsilon_{init}^\kit(\omega)$ and $\epsilon_{opt}^\kit(\omega)$ are non-zero, since model reduction and model matching are not exact. A consequence of~\eqref{eq.ImprovCond} is that non-zero error terms can even be beneficial for the success of the optimization. If $\epsilon_{opt}^\kit(\omega)$ is smaller than zero, and $\epsilon_{init}^\kit(\omega)$ is greater than zero, for all $\omega \in \R$,~\eqref{eq.ImprovCond} becomes less restrictive than~\eqref{eq.ReducedImprovCond}. However, this would mean that the reduced model $\tG$ needs to overestimate the detailed model $G$ for the initial parameters $\vK_{init}^\kit$ and $\vtK_{init}^\kit$, and vice-versa for the optimized parameters $\vK_{opt}^\kit$ and $\vtK_{opt}^\kit$. Finding a model and optimization procedure, which guarantees this property, is in general not possible. Hence, we adopt the strategy to find a model which achieves minimal error terms both before and after the optimization. 

\subsection{Evaluation of the improvement of the system $\Hinf$ norm}

To evaluate~\eqref{eq.ImprovCond}, $\epsilon_{init}^\kit(\omega)$ and $\epsilon_{opt}^\kit(\omega)$ need to be known. This is, however, challenging, as $G(\vK,s)$ is unknown and, thus,~\eqref{eq.RedErr} and~\eqref{eq.MatchErr} cannot be directly evaluated.	
This can be overcome by requiring that each SO sends the discrepancies at sampling frequencies to the SC
\begin{align}
\Delta_{i}(\omega)^\kit = G_i(\vK_{i,opt}^\kit, \jw) - \tG_i^\kit(\vtK_{i,opt}^\kit, \jw). \label{eq.ReductionError}
\end{align}
As the SC knows $\tG_i^\kit(\vtK_{i,opt}^\kit, \jw)$, the values of $G(\vK_{i,opt}^\kit,\jw)$ can be calculated combining~\eqref{eq.ReductionError} and~\eqref{eq.DetailedCoupledSys} for frequency samples, and consequently $\HinfNorm{G(\vK_{i,opt}^\kit,\jw)}$, without knowledge of the detailed structure of each $G_i$. The Condition~\eqref{eq.ImprovCond} can be directly used to evaluate whether the optimization was successful.

\begin{rem}
	The sampled values $\Delta_{i}(\omega)^\kit$ do not explicitly reveal detailed information about the structure and parameters of the subsystems. However, arguably, providing such sampling information decreases the data privacy.
	By using matrix norm inequalities, it is possible to find sufficient conditions for the norm improvement which require less data to be exchanged between the SOs and SC. One example is $\epsilon_{opt}^\kit (\omega)$. However, this introduces conservativeness to the approach, making it less applicable, whereas~\eqref{eq.ReductionError} allows us to check the norm-improvement without any conservativeness. Finding less conservative conditions with increased data privacy is a part of future research. 
\end{rem}

\subsection{Detailed algorithm for the proposed approach}

The overall proposed algorithm is shown in Fig.~\ref{fig.AlgDetailed}.
\begin{figure}[tb]
	\centering
	\includegraphics[width=1\columnwidth]{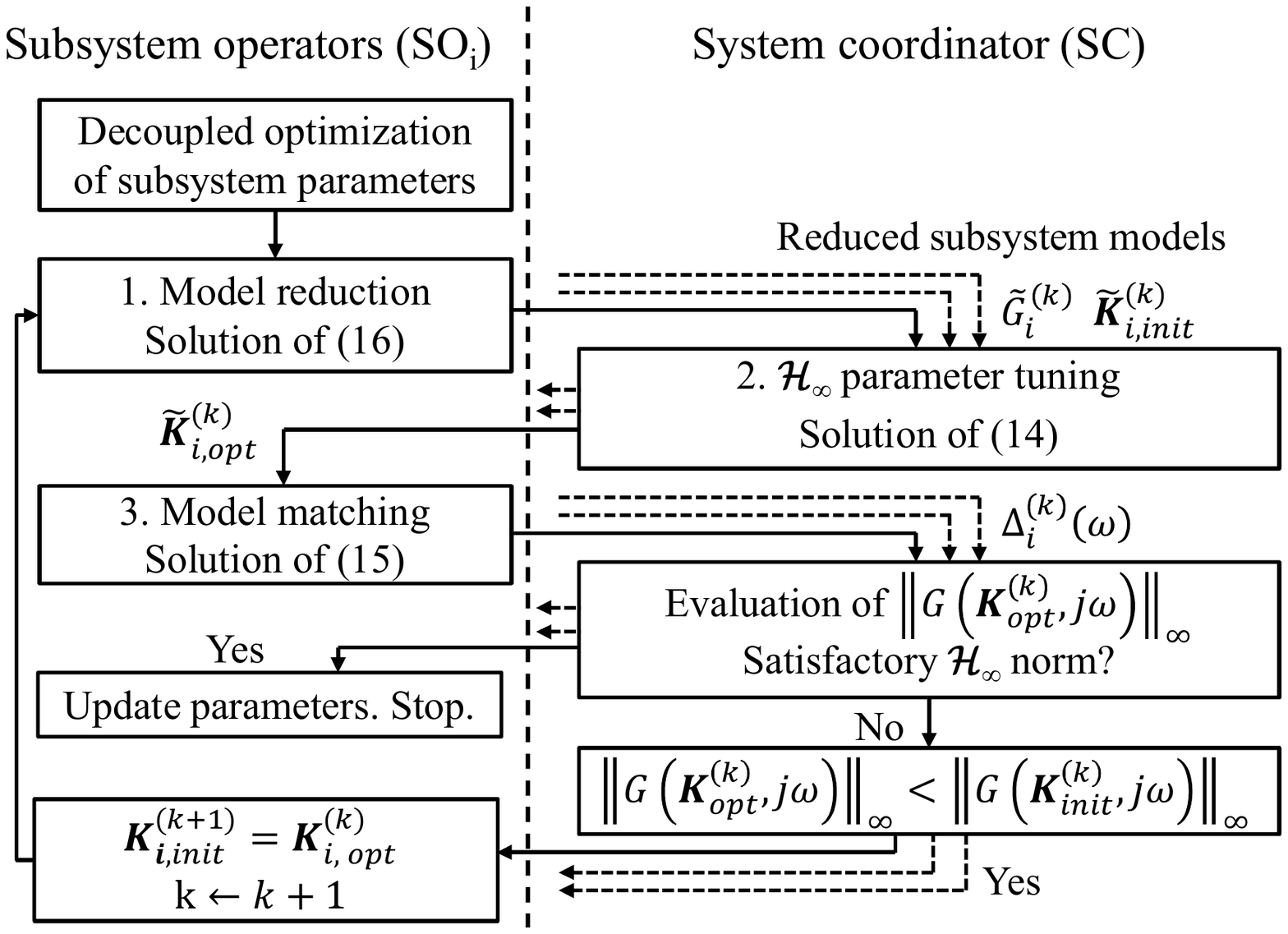}
	\caption{Detailed illustration of the proposed optimization approach.}
	\label{fig.AlgDetailed}
\end{figure}
Before the iterations, each SO first tunes its parameters in order to eliminate local oscillations. As described in Section~\ref{subsubsec.SubsysModel}, we use the transfer function $G^*_i(\vK_i,s)$, defined with~\eqref{eq.TFOmegas} for this purpose, where the coupling inputs are used as disturbance inputs for the optimization, in addition to the infeeds of static prosumers. Afterwards, the iterative procedure can start.

In each iteration, the SOs reduce their detailed area model $G_i(\vK_i,\jw)$ by optimizing the parameters of the structured reduced model $\tG_i'(\vR_i,\vtK_i,\jw)$ with~\eqref{eq.ModelReductionProblem}. The reduced models, together with the initial reduced parameter vectors $\vtK_{i,init}^\kit$ are then sent to the SC. The SC optimizes the reduced system model with~\eqref{eq.StructReducedOpt} and sends the optimized parameter vectors $\vtK_{i,opt}^\kit$ to the SOs. Each SO performs the model matching step by solving~\eqref{eq.ModelMatching}, and sends $\Delta_{i}^\kit(\jw)$ to the SC. 

The SC calculates $\HinfNorm{G(\vK_{opt}^\kit,\jw)}$ with~\eqref{eq.ReductionError} and~\eqref{eq.DetailedCoupledSys}.
Based on the result, the SC can make the decision to stop the algorithm if the results are satisfactory, or to continue with the next iteration.

\section{Application to the IEEE 68 bus example}
\label{sec.IEEE68bus}

\begin{figure}[tb]
	\centering
	\includegraphics[width=1\columnwidth]{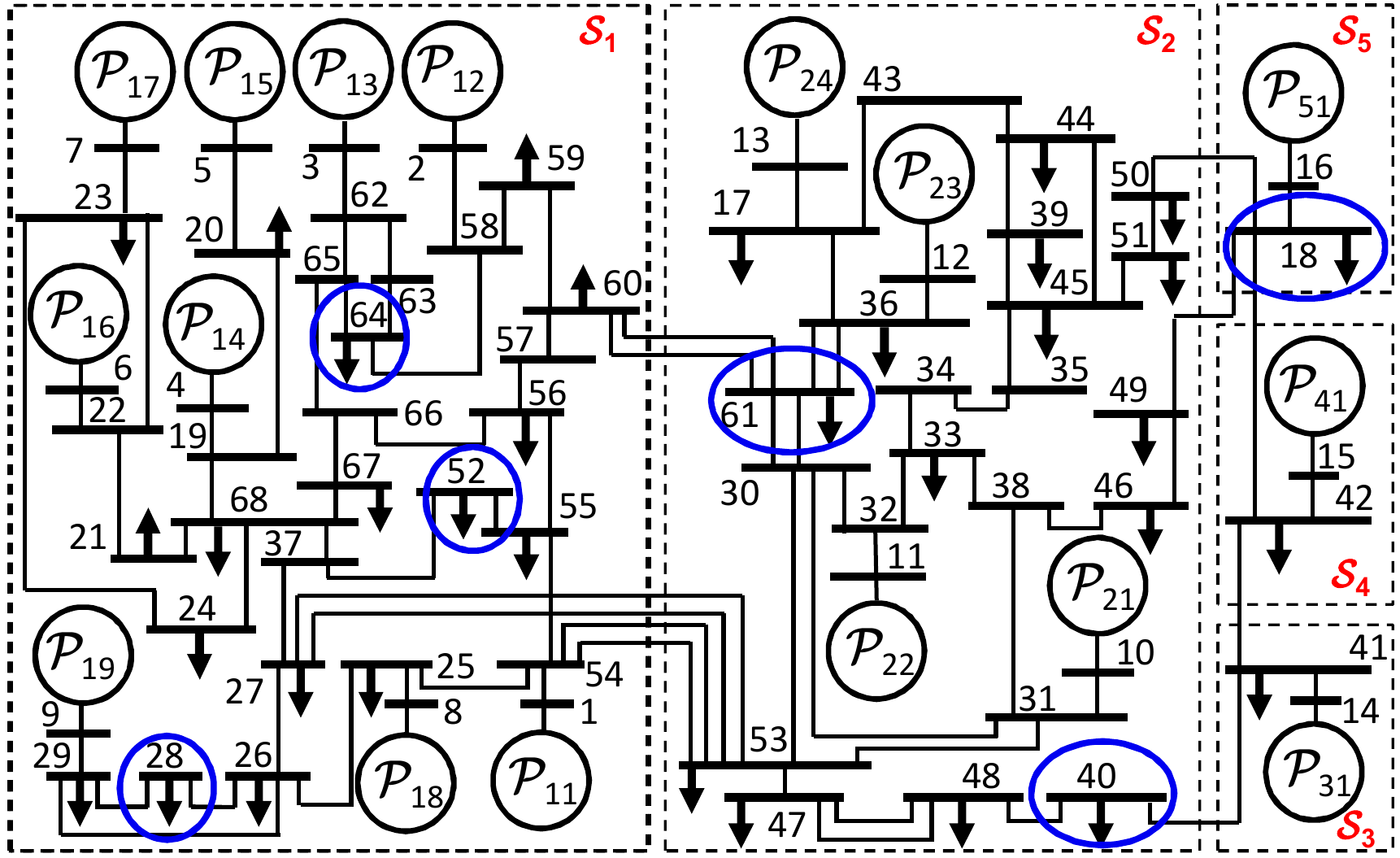}
	\caption{IEEE 68 bus model with 16 power plant prosumers, divided into 5 subsystems~\cite{singh2013ieee}. The disturbance inputs considered for the $\Hinf$ parameter tuning are marked in blue.}
	\label{fig.IEEE68}
\end{figure}
We consider the IEEE 68 bus power system~\cite{singh2013ieee} with 16 dynamic prosumers (power plants), shown in Fig.~\ref{fig.IEEE68}. The system consists of five subsystems $\cS_i$, i = 1...5, coupled with power lines, where the subsystems 3-5 are represented by reduced models. The parameters of the power grid and synchronous generators are provided in~\cite{singh2013ieee}. All generators are operated with standard IEEE controllers, see~\ref{App.IEEE68Models}.

In order to justify the motivation of online parameter tuning for power systems with high shares of renewable generation, we first apply the centralized tuning method, which also allows to show the effects of increased renewable penetration in power systems. Subsection~\ref{subsec.ApplicationDistributed}, results for the hierarchical algorithm. 

\subsection{Centralized tuning and the impact of renewables}
\label{subsec.TimeVarDynamics}
In power systems, oscillations with a damping ratio below 5\% are considered weakly dampened~\cite{dampratio}.
With the initial parameterization of the controllers, eight oscillatory modes show damping ratios below 5\%, see Table.~\ref{tab.InitOscillModes}. The step response of the system to a 100 MW load step in bus 52 is shown in Fig.~\ref{fig.IEEE68Init}. The simulation is done with a linearized system model, which has shown good accuracy in previous works~\cite{Mesanovic18ACC, mesanovic2018optimalparameter,Mesanovic17ISGT,mesanovic2018ISGT}. The linear system has 280 states and 160 controller parameters with the initial parameter vector $\vK_{init}$.

\begin{table}[tb]
	\centering
	\caption{Modes of the IEEE 68 bus system which have damping ratios below 5\% for the initial controller parameterization.}
	\label{tab.InitOscillModes}
	\begin{tabular}{c c c}
		\toprule
		mode & frequency (rad/s) & damping ratio (\%) \\ \midrule
		1   & 3.8               & 0.5                \\
		2   & 2.5               & 1.7                \\
		3   & 3.5               & 2.3                \\
		4   & 6.3               & 2.9                \\
		5   & 5                 & 4.3                \\
		6   & 7.5               & 4.4                \\
		7   & 6.3               & 4.5                \\
		8   & 8.2               & 4.6                \\ \bottomrule
	\end{tabular}
\end{table}
\begin{figure}[tb]
	\centering
	\includegraphics[width=1\columnwidth]{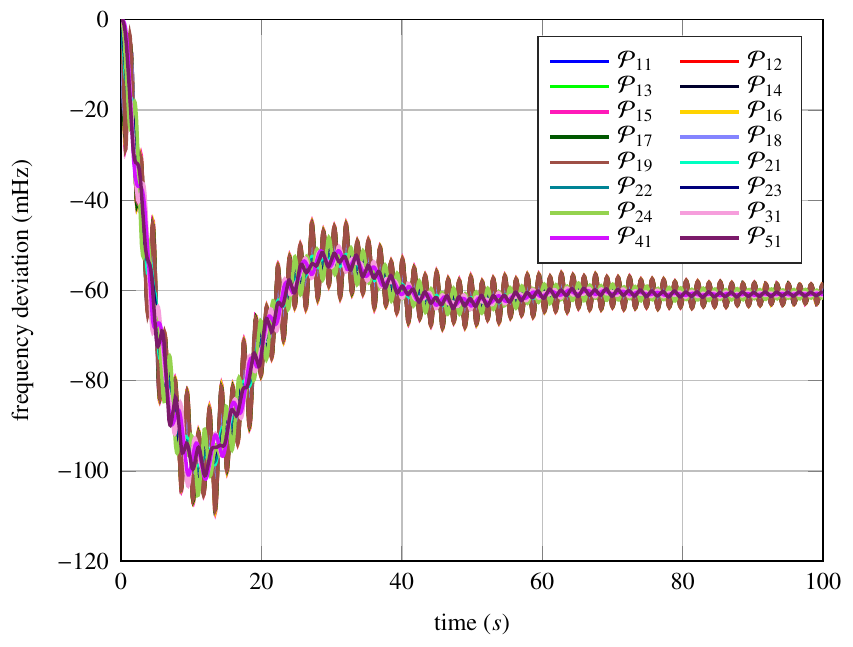}
	\caption{Frequency response for a 100 MW load step in bus 52. Power plants in subsystems 2 and 3 oscillate against areas 1 and 5.}
	\label{fig.IEEE68Init}
\end{figure}
\begin{figure}[tb]
	\centering
	\includegraphics[width=1\columnwidth]{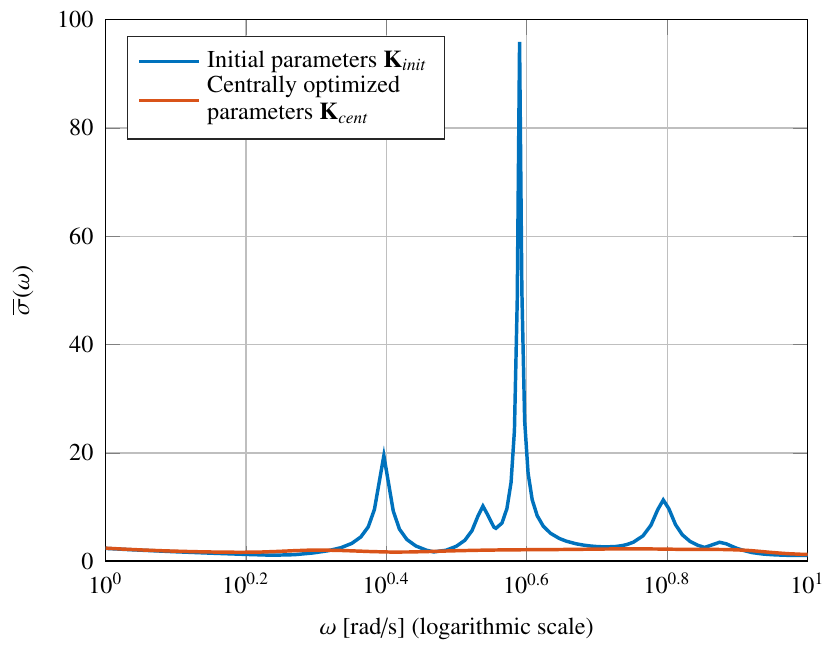}
	\caption{Largest singular value of the IEEE 68 bus system as a function of frequency $\omega$. After centralized optimization with~\eqref{eq.FiniteOptProb}, the resonant peaks in the system are practically eliminated.}
	\label{fig.IEEE68Sigma}
\end{figure}
\begin{figure}[tb]
	\centering
	\includegraphics[width=1\columnwidth]{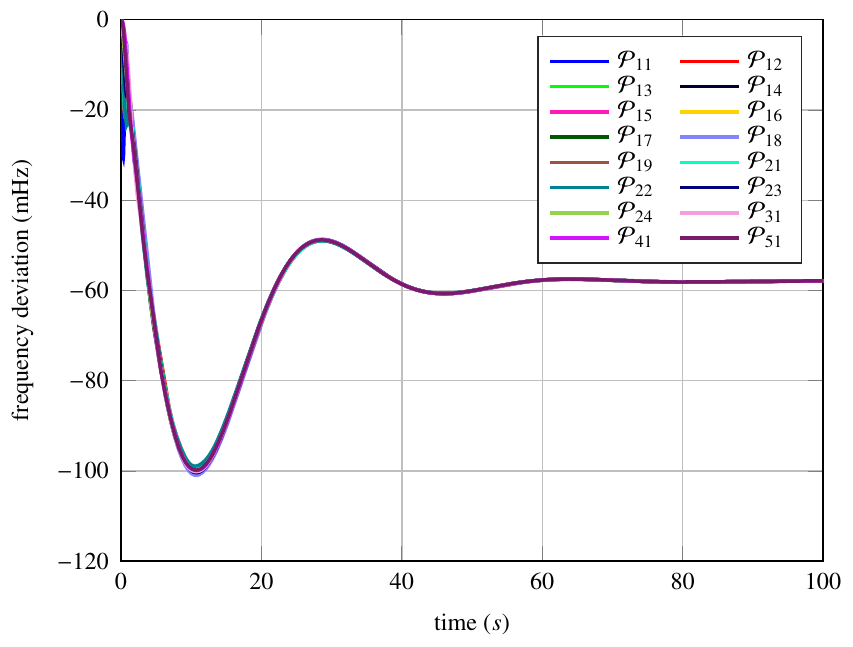}
	\caption{Frequency response for a 100 MW load step in bus 52 for the centrally tuned parameters $\vK_{cent}$.}
	\label{fig.IEEE68Opt}
\end{figure}

For the centralized tuning approach presented in Section~\ref{sec.CentTuning} we assume that the transfer function of the entire system is completely known. The solution of the centralized problem will serve as a baseline for the subsequent hierarchical optimization.
For the optimization, we use the Matlab toolbox YALMIP~\cite{Yalmip}, together with the solver SeDuMi~\cite{Sedumi}, and the transfer function $G^*$ from~\eqref{eq.TFOmegas}.
Disturbance inputs are the active power infeeds of the static prosumers, see Fig.~\ref{fig.IEEE68}.
After tuning, all oscillatory modes are sufficiently dampened with a damping  ratio above 7\%. We denote the centrally tuned parameter vector of all prosumers with $\vK_{cent}$.
The system $\Hinf$ norm was reduced by 97.5\%. Figure~\ref{fig.IEEE68Sigma} shows the largest singular values of the system for the initial and optimized parameters. The largest values of the curves in Fig.~\ref{fig.IEEE68Sigma} represent the $\Hinf$ norm with the initial and optimized parameters. 
The largest peak at approx. 3.8 rad/s corresponds to the oscillatory mode with the poorest damping ratios in Table~\ref{tab.InitOscillModes}. With $\vK_{cent}$, the peak is completely eliminated, which is also visible in Fig.~\ref{fig.IEEE68Opt}.

When the amount of renewable generation in the system increases, conventional prosumers, such as power plants, will be disconnected from the grid to prevent overproduction. In the considered power system, power plants already represent aggregated models of multiple smaller power plants, thus, the disconnection of smaller power plants is modeled by reducing the nominal power of the power plants $\cP_{ij}$.
The power infeed of the power plant is also reduced and shifted to static prosumers in the same bus, i.e., we model renewable generation in this system as static prosumers~\cite{poolla2019placement,pddotnuschel2018frequency}. 

To simulate the effects of large scale renewable generation, we consider two scenarios.
The first scenario considers increased renewable integration in $\cS_3$. For this purpose, we reduce the power of $\cP_{31}$ to 15\% of its original value and assume that power plants with PSSs are disconnected. In the second scenario, we consider increased renewable generation in $\cP_{22}$. We rescale the power of $\cP_{22}$ to 50\% of its original value and deactivate PSS$_{22}$. Details of the scenarios are presented in Table~\ref{tab.ScenariosAndModes}. In both scenarios, the damping becomes worse for the initial parameters $\vK_{init}$. Even though $\vK_{cent}$ eliminates weakly dampened eigenmodes for the initial, nominal, scenario, weakly dampened oscillations still emerge in the other two scenarios. New parameter sets are needed to improve the oscillation damping for these scenarios, demonstrating the necessity for online adaptation of parameters to counteract the change.

\begin{table}[tb]
	\centering
	\caption{Frequency and damping of three weakest-dampened modes, which have damping ratios below 5\%, for the considered scenarios and controller parameters.}
	\label{tab.ScenariosAndModes}
	\begin{tabular}{c c c}
		\toprule
		Scenario      & $\vK_{init}$                                                           & $\vK_{cent}$                                                      \\ \midrule
		Initial      & \makecell{3.9 rad/s, 0.5\% \\ 2.51 rad/s, 1.7\% \\ 6.28 rad/s, 2.9\%}  &                                                                   \\
		$\cP_{31}$ scaled & \makecell{3.83 rad/s, 0.2\% \\ 2.7 rad/s, 1.4\% \\ 6.28 rad/s, 2.9\%}  & 4.02 rad/s, 3.6\%                                                 \\
		$\cP_{22}$ scaled & \makecell{3.9 rad/s, 0.4\% \\ 2.45 rad/s, 1.7\%  \\ 6.28 rad/s, 2.9\%} & \makecell{2.51 rad/s, 3\%\\ 3.46 rad/s, 3\% \\ 4.9 rad/s, 3.7\% } \\ \bottomrule
	\end{tabular}
\end{table}


\subsection{Hierarchical data privacy conserving $\Hinf$ parameter tuning}
\label{subsec.ApplicationDistributed}


We now apply the proposed hierarchical data privacy preserving tuning method on the IEEE 68 bus power system. Note that as subsystems $\cS_3$ - $\cS_5$ are already replaced with reduced models, the reduction and matching step is not necessary for these systems.
\begin{figure}[tb]
	\centering
	\includegraphics[width=1\columnwidth]{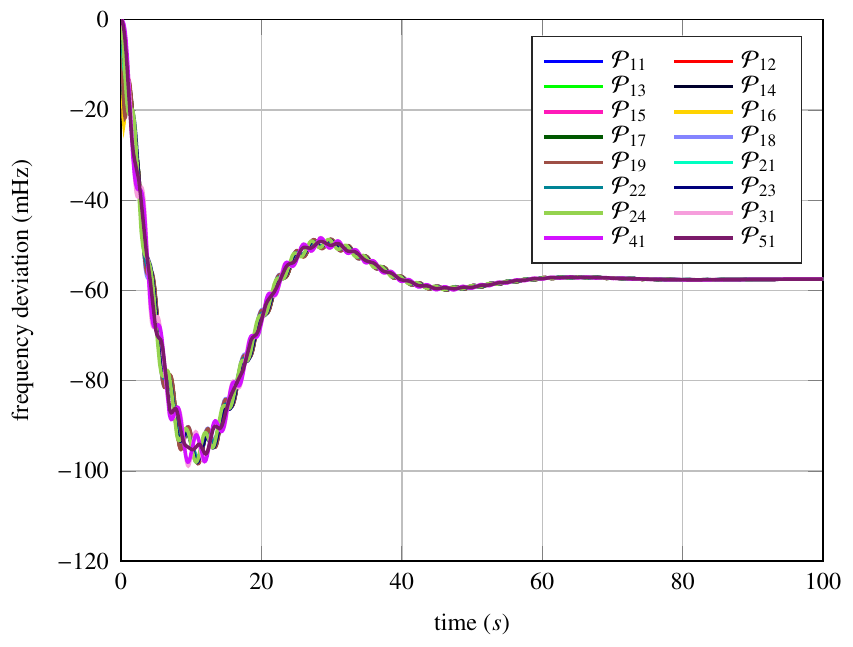}
	\caption{Frequency response of the detailed model to a 100 MW load step in bus 52 for completely decentrally optimized parameters $\vK_{dcp}$.}
	\label{fig.SeparateOptDetailed}
\end{figure}

As shown in Fig.~\ref{fig.AlgDetailed}, each SO first tunes the parameters of its controllers without regarding other systems by using $G_i^*(\vK_i, s)$. We consider thereby the coupling inputs as disturbances. Figure~\ref{fig.SeparateOptDetailed} shows the step response of the detailed coupled system after this step. We denote this (decoupled) parameterization of each subsystem with $\vK_{dcp} = \myvec{i}(\vK_{dcp,i})$. As can be seen, the step response was significantly improved compared to the response with initial parameters in Fig.~\ref{fig.IEEE68Init}. However, it is still worse than the centrally tuned parameterization $\vK_{cent}$ in Fig.~\ref{fig.IEEE68Opt}, and three weakly dampened modes still remain, as summarized in Table~\ref{tab.ModesAfterSeparateOpt}. Thus, even though each SO eliminated the oscillations within the subsystem, oscillations betweeen the subsystems could not be eliminated without consideration of the coupling between the subsystems.

\begin{table}[tb]
	\centering
	\caption{Frequency and damping ratios of weakly-dampened modes in the coupled system when each SO optimizes its parameters separately, i.e. with $\vK_{dcp}$.}
	\label{tab.ModesAfterSeparateOpt}
	\begin{tabular}{c c c}
		\toprule
		Mode & frequency (rad/s) & damping ratio (\%) \\ \midrule
		1   & 2.51              & 2.6          \\
		2   & 3.52              & 3.8          \\
		3   & 4.9               & 4.4          \\ \bottomrule
	\end{tabular}
\end{table}

In order to improve the step response, we apply the described optimization steps in the next sections.

\subsection{Structured model reduction}
\label{subsec.StructRed4PowSys}



Current state of the art approaches for structured model reduction are not suitable for our approach.
First, the approaches require the exchange of detailed parameters/data of the subsystems, which violates our goal of data privacy.
Second, unstructured models are often generated, which do not allow to retain insight into what parameters can be tuned to minimize the $\Hinf$ error between the reduced and detailed model.

For this reason, we introduce a hybrid approach for distributed model reduction for power systems. It consists of two steps. In the first step, we create an equivalent model of a synchronous generator (SG) from a group of SGs in $\cS_i$. As equivalent models of SGs are well studied in the literature, we use the analytical procedure from~\cite{chow1995inertial} for this step. 


In the second step, we parameterize of equivalent controllers for the SG. This challenge, however, was not intensively studied.
Thus, to parameterize the controllers of the equivalent SG, i.e. TGOV$_{ij}$, AVR$_{ij}$, and PSS$_{ij}$, we use~\eqref{eq.ModelReductionProblem}. We first choose models for the controllers occurring in the system with the highest frequency. In case of the IEEE 68 bus system, those are the controllers in Figs.~\ref{fig.TGOV},~\ref{fig.Exciter} and~\ref{fig.PSS} in~\ref{App.IEEE68Models}. The tunable controller parameters, which are part of the vector $\vtK_{i}$, are marked red in the figures, whereas all other parameters are a part of the vector $\vR_i$.

In case of the IEEE 68 bus system, subsystems $\cS_1$ and $\cS_2$ are each replaced by one equivalent power plant, which has shown to be sufficient. Figures~\ref{fig.ReductionS1} and~\ref{fig.ReductionS2} show the largest singular values of the reduced and detailed model of $\cS_1$ and $\cS_2$, respectively, before and after the parameterization procedure. They show a very good match between the detailed models and the reduced models after the parameterization.

\begin{figure}[tb]
	\centering
	\includegraphics[width=0.95\columnwidth]{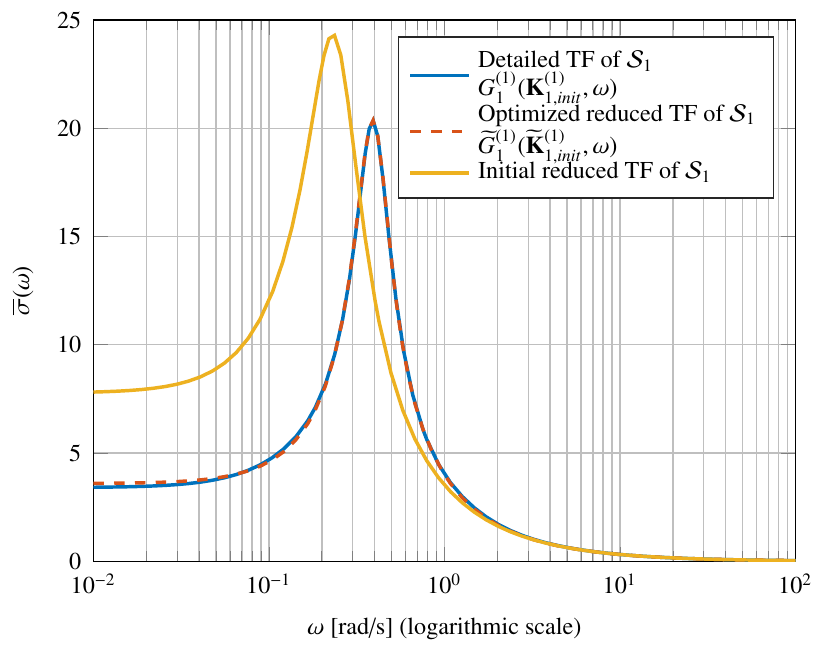}
	\caption{Largest singular values of the detailed and reduced model of $\cS_1$. Note that the blue curve is mostly covered by the red curve.}
	\label{fig.ReductionS1}
\end{figure}
\begin{figure}[tb]
	\centering
	\includegraphics[width=0.95\columnwidth]{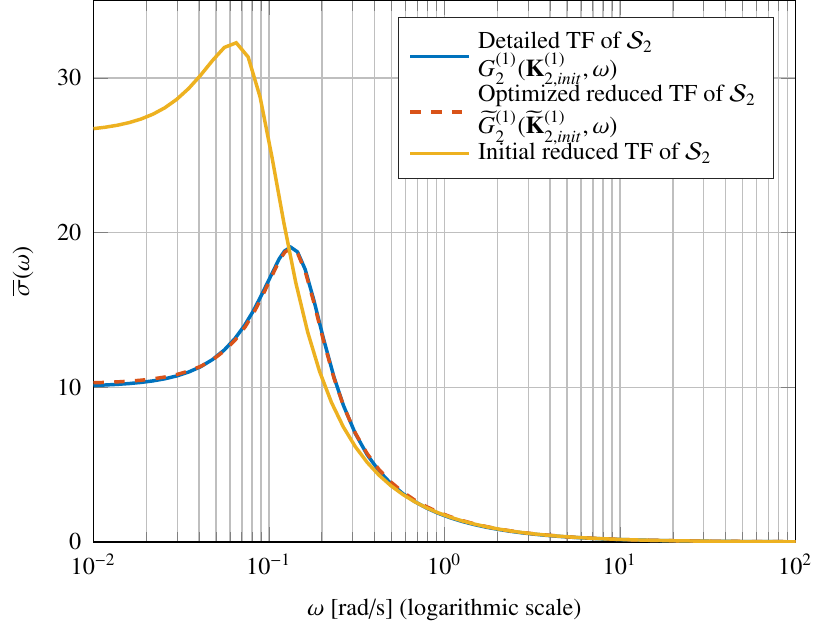}
	\caption{Largest singular values of the detailed and reduced model of $\cS_2$, showing a good match. Note that the blue curve is mostly covered by the red curve.}
	\label{fig.ReductionS2}
\end{figure}

\subsection{Structured $\Hinf$ optimization for the overall reduced system}

The SC lumps the reduced subsystem models into $\tG(\vtK, s)$. The reduced system for the IEEE 68 bus model is depicted in Fig~\ref{fig.ReducedSys}. The static infeeds, marked blue, are used as disturbance inputs, and are elements of the reduced vector of disturbances $\vtw_S$.
A system with 89 states and 50 optimization parameters is obtained, whereas the detailed model has 280 states and 160 controller parameters. 
We denote the initial parameter vector of the detailed and the reduced model by $\vK_{init}^{(1)}$ and $\vtK_{init}^{(1)}$, respectively.
The step response of the reduced system to a 100 MW load step in bus 2 is shown in Fig.~\ref{fig.StepResponseReduced}. The oscillations in the coupled system are less dampened than with $\vK_{init}^{(1)}$ in the detailed model, c.f. Fig.~\ref{fig.SeparateOptDetailed}. However, a 100\% accuracy of the reduced model is not required for the approach to be successful, as described in Section~\ref{subsec.ImprovCond}. 
\begin{figure}[tb]
	\centering
	\includegraphics[width=0.5\columnwidth]{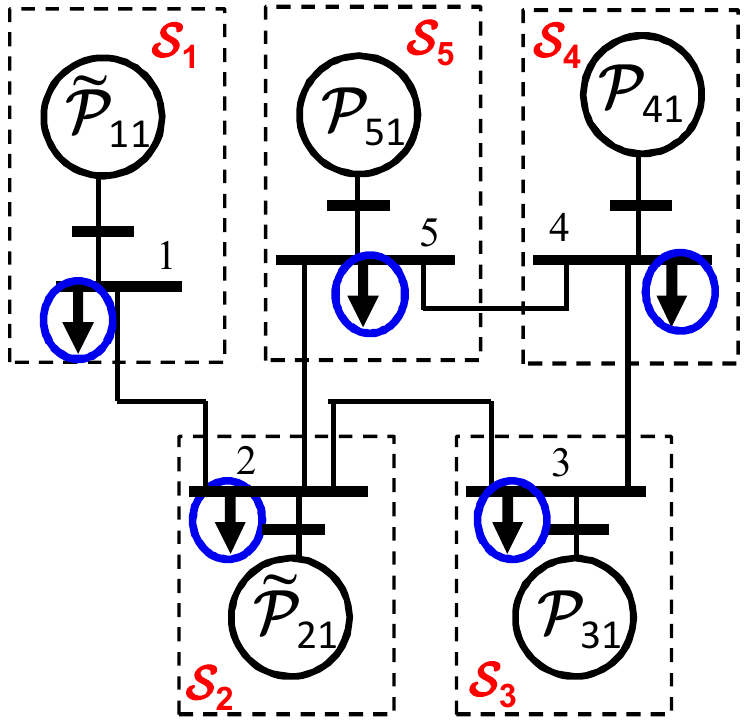}
	\caption{Reduced model of the IEEE 68 bus power system. Disturbance inputs for the optimization are marked blue.}
	\label{fig.ReducedSys}
\end{figure}
\begin{figure}[tb]
	\centering
	\includegraphics[width=1\columnwidth]{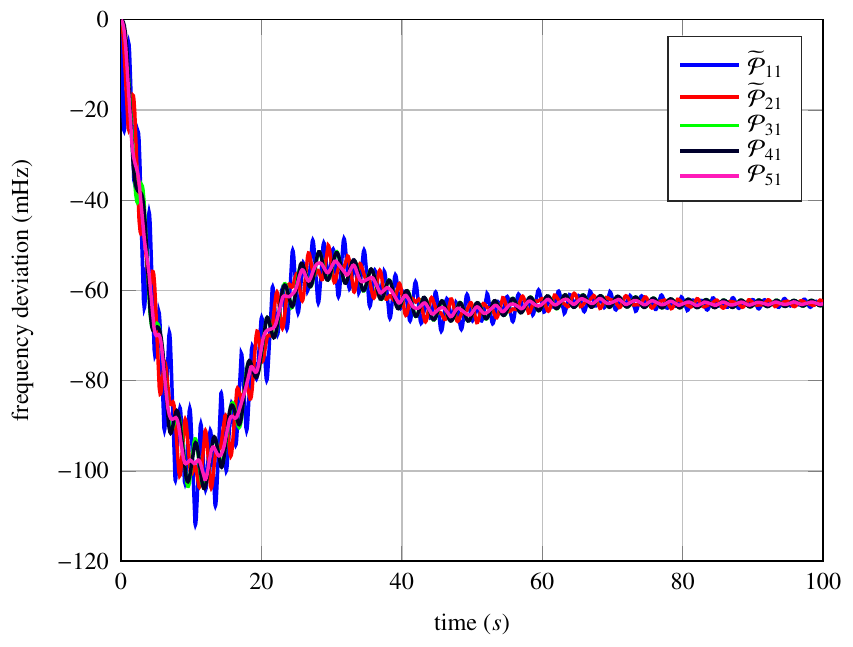}
	\caption{Step response of the reduced model in Fig.~\ref{fig.ReducedSys} in the first iteration $\tG^{(1)}(\vtK_{init}^{(1)},\omega)$ to a 100 MW load step in bus 2 with  $\vtK_{init}^{(1)}$.}
	\label{fig.StepResponseReduced}
\end{figure}

The SC optimizes the reduced model using~\eqref{eq.StructReducedOpt}, and obtains the optimized parameter vector $\vtK_{opt}^{(1)}$. 

\subsection{Model matching}

In the last step, the SOs optimize with~\eqref{eq.ModelMatching} the parameters of the detailed subsystem models to match the reduced model. Figures~\ref{fig.MatchingS1} and~\ref{fig.MatchingS2} show the results of the model matching step for $\cS_1$ and $\cS_2$. For $\cS_2$, model matching could not decrease the difference between the detailed and reduced model.
However, the error can be reduced by relaxing the box constraints of controller parameters in $\cS_2$ or by making the box constraints for $\vtK_{2}$ in the reduced model tighter, which was not necessary, as the obtained coupled system norm is sufficiently good even after one iteration.

Counter-intuitively, the reduced reference models for $\cS_1$ and $\cS_2$ have lager $\Hinf$ norms than the initial models. This underpins that $\Hinf$ optimization of the decoupled subsystems does not necessarily minimize the $\Hinf$ norm of the coupled system. This is also evident from the optimization results shown subsequently in Fig.~\ref{fig.SigmaAll}, in which the optimized coupled (detailed and reduced) systems have a lower $\Hinf$ norm than the initial respective systems, even though the $\Hinf$ norm of the decoupled subsystems increased.

\begin{figure}[tb]
	\centering
	\includegraphics[width=1\columnwidth]{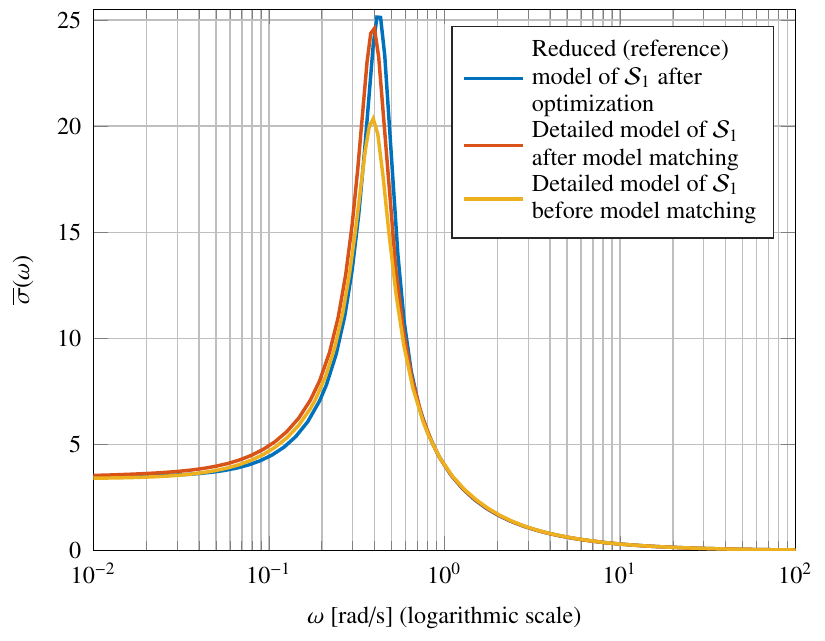}
	\caption{Largest singular values of the detailed and reduced system for $\cS_1$ before and after the model matching step.}
	\label{fig.MatchingS1}
\end{figure}

\begin{figure}[tb]
	\centering
	\includegraphics[width=1\columnwidth]{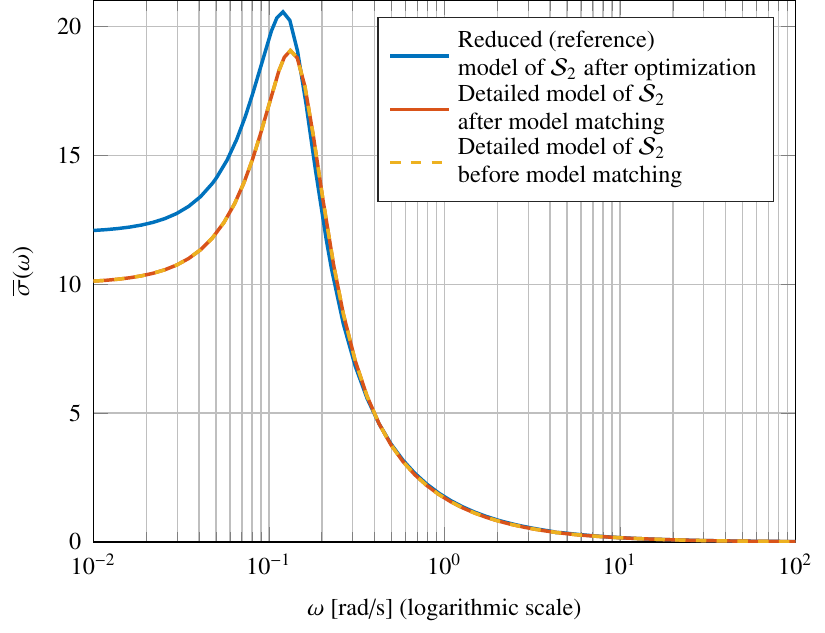}
	\caption{Largest singular values of the detailed and reduced system for $\cS_2$ before and after the model matching step. Note that the red curve is under the yellow curve.}
	\label{fig.MatchingS2}
\end{figure}

\subsection{Results of hierarchical $\Hinf$ controller tuning}


After the model matching step, the parameter vectors $\vK_{i,opt}^{(1)}$ are obtained for each subsystem $\cS_i$. The time response of the detailed system, with the parameterization $\vK_{opt}^{(1)} = \myvec{i}(\vK_{i,opt}^{(1)})$, to a 100 MW load step in bus 52 is shown in Fig.~\ref{fig.OptimizedReducedDetailed} with solid lines. The same figure also shows the optimized step response of the reduced model with $\vtK_{opt}^{(1)}$ in dashed lines, showing a very good correspondence of the two models.
Furthermore, the response looks almost identical as the results with centralized tuning in Fig.~\ref{fig.IEEE68Opt}. With $\vK_{opt}^{(1)}$, all weakly-dampened modes are eliminated from the system as well. The system $\Hinf$ norm was reduced by 97.5\%, which was also the case with $\vK_{cent}$. This could not be achieved with $\vK_{dcp}$, when all SOs tuned their controller parameters separately.

Figure~\ref{fig.SigmaAll} shows the largest singular values of the detailed and reduced system in the relevant frequency range for the various parameterizations. With the decoupled parameterization $\vK_{dcp}$, a large peak is still present in the system at approx 2.5 [rad/s]. The reduced model with $\vtK_{init}^{(1)}$ is able to recreate this peak, and introduces an additional peak at approx. 4 Hz. However, 100\% accuracy of the reduced model is not necessary for the approach to be successful. Both peaks in the reduced model are eliminated with $\vtK_{opt}^{(1)}$. The resulting parameterization of the detailed model $\vK_{opt}^{(1)}$ with the decentralized approach achieves approximately the same results as the centralized parameterization $\vK_{cent}$.

\begin{figure}[tb]
	\centering
	\includegraphics[width=1\columnwidth]{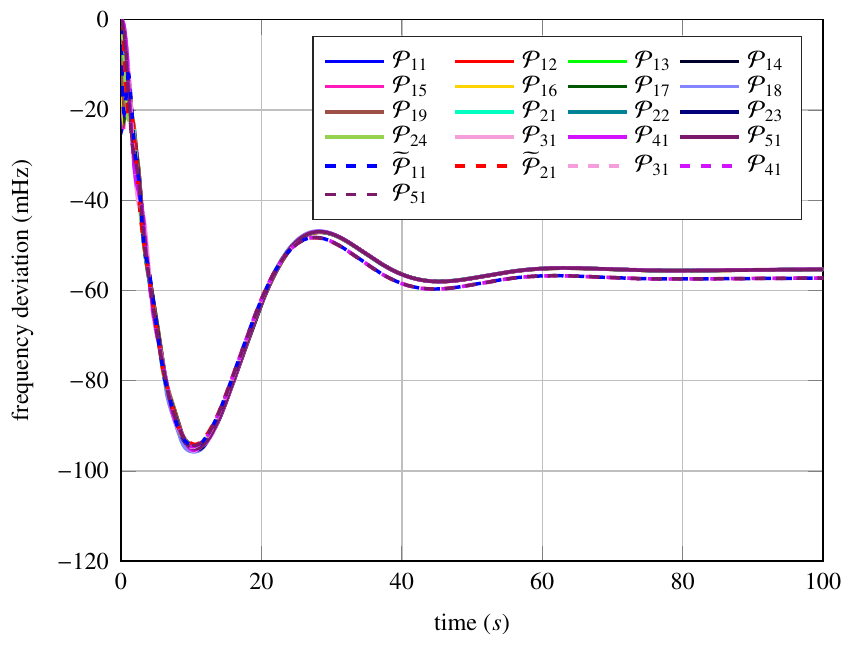}
	\caption{Step response of the optimized detailed and reduced IEEE 68 bus model to a 100 MW load step. Solid lines represent simulation of a 100 MW load step in bus 52 with the detailed model and  $\vK_{opt}^{(1)} = \myvec{i}(\vK_i^{(1)})$, whereas dashed lines represent simulation of a 100 MW load step in bus 2 with the reduced model and  $\vtK_{opt}^{(1)}$.}
	\label{fig.OptimizedReducedDetailed}
\end{figure}

\begin{figure}[tb!]
	\centering
	\includegraphics[width=1\columnwidth]{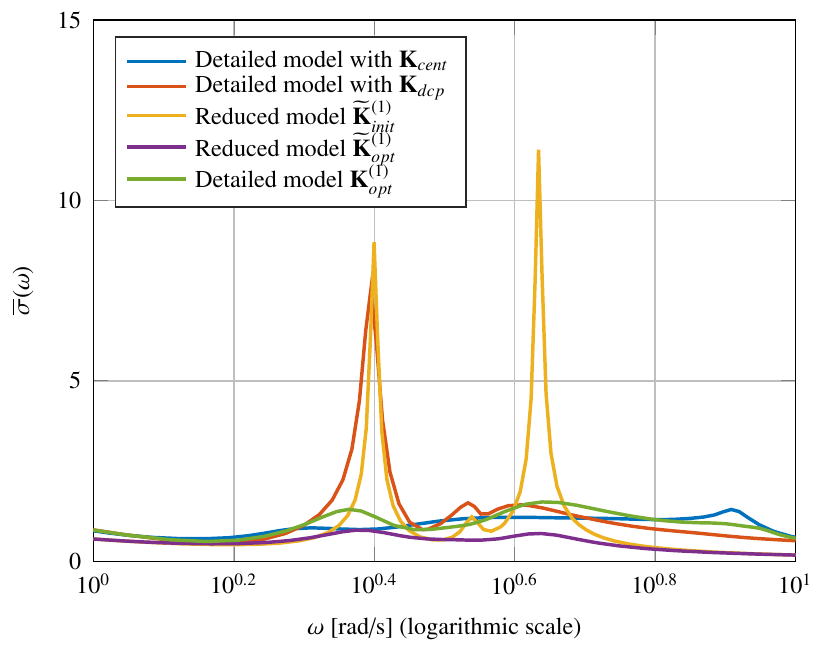}
	\caption{Largest singular values of the reduced and detailed models of the IEEE 68 bus system achieved with the various parameterizations.}
	\label{fig.SigmaAll}
\end{figure}



\section{Tuning for an artificial large scale power system}
\label{sec.ApplicationLarge}

To show the applicability of the proposed approach on a large and complex power system, we couple nine IEEE 68 bus power systems, as shown in Fig.~\ref{fig.LargeSystem}. The northern connection in each subsystem is made with bus 43, the southern connection with bus 48, the eastern with 42, and western with 21.
This system is strongly meshed, and has 2520 states and 1440 controller parameters, making centralized optimization impossible.

\begin{figure}[tb]
	\centering
	\includegraphics[width=1\columnwidth]{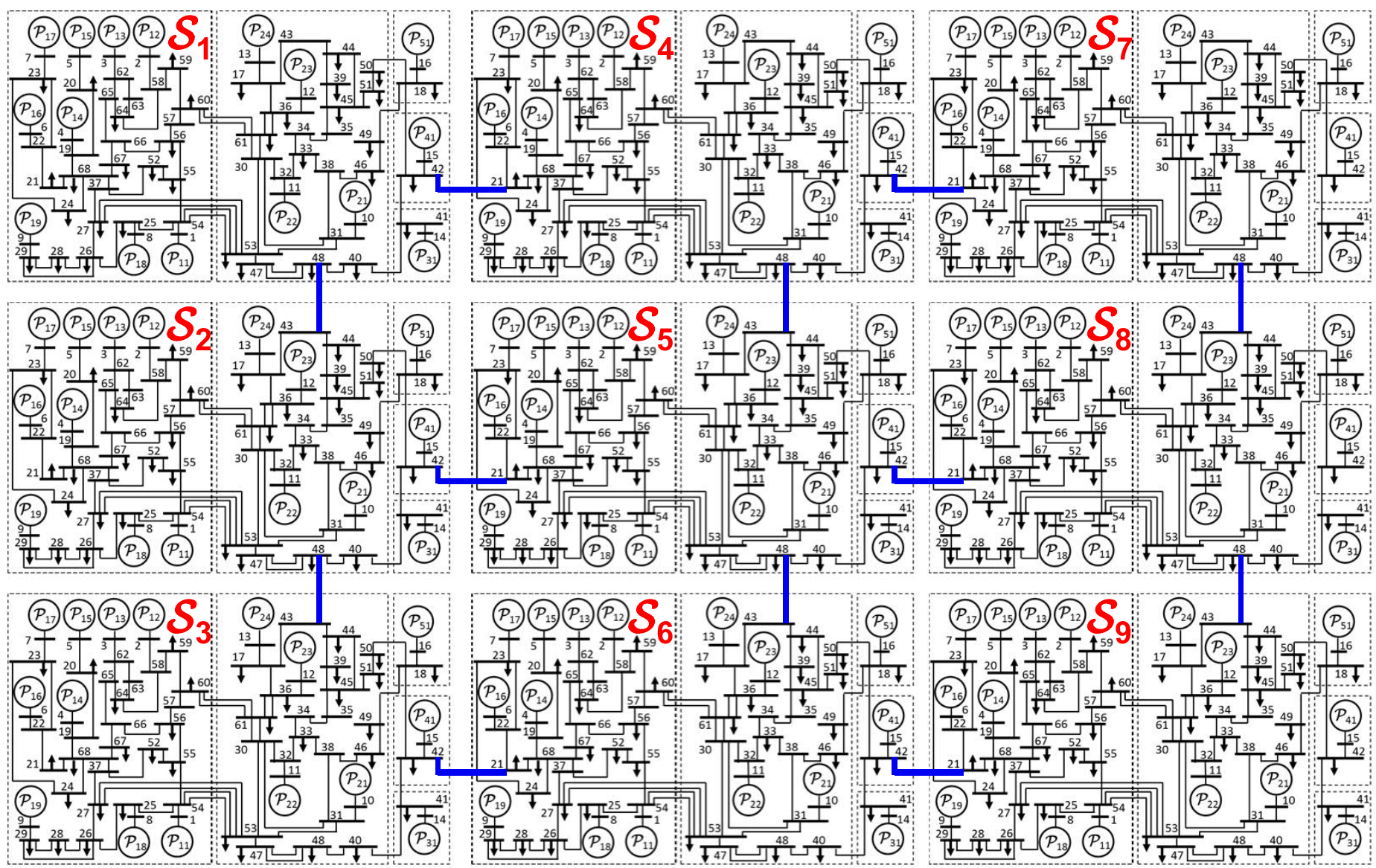}
	\caption{Large synthetic power system obtained by coupling nine IEEE 68 bus systems. The system consists of 144 power plants, 2520 dynamic states and 1440 controller parameters.}
	\label{fig.LargeSystem}
\end{figure}

\begin{figure}[tb]
	\centering
	\includegraphics[width=1\columnwidth]{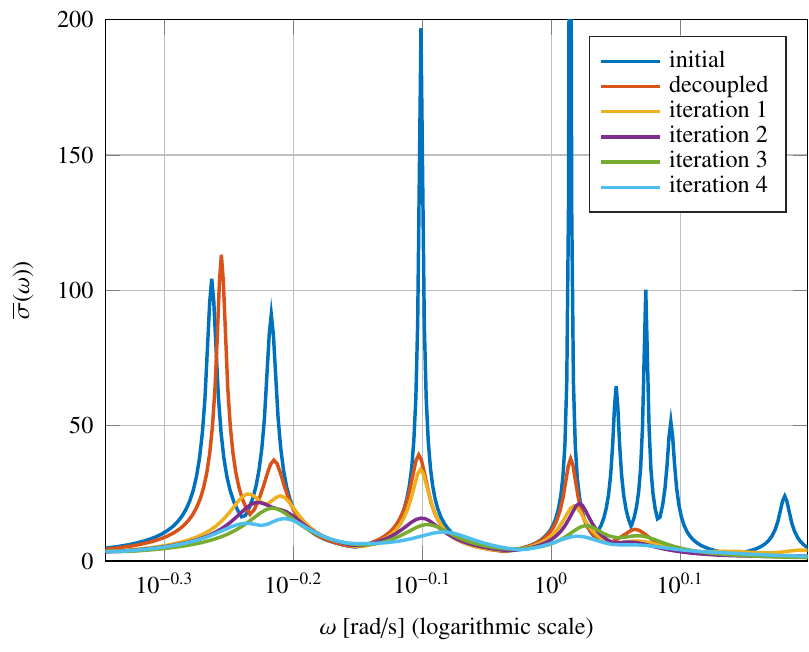}
	\caption{Largest singular values of the large system in Fig.~\ref{fig.LargeSystem} with the initial parameterization, decoupled parameterization, when each SO optimizes its parameters separately, and after the model matching step in each iteration of the optimization. The highest peak of the plot with initial parameterization, corresponding to the value 360 at approximately 1 rad/s, is cut-off to improve visibility of other plots.}
	\label{fig.BisSys_sigma}
\end{figure}

%

\begin{figure}[tb]
	\centering
	\includegraphics[width=1\columnwidth]{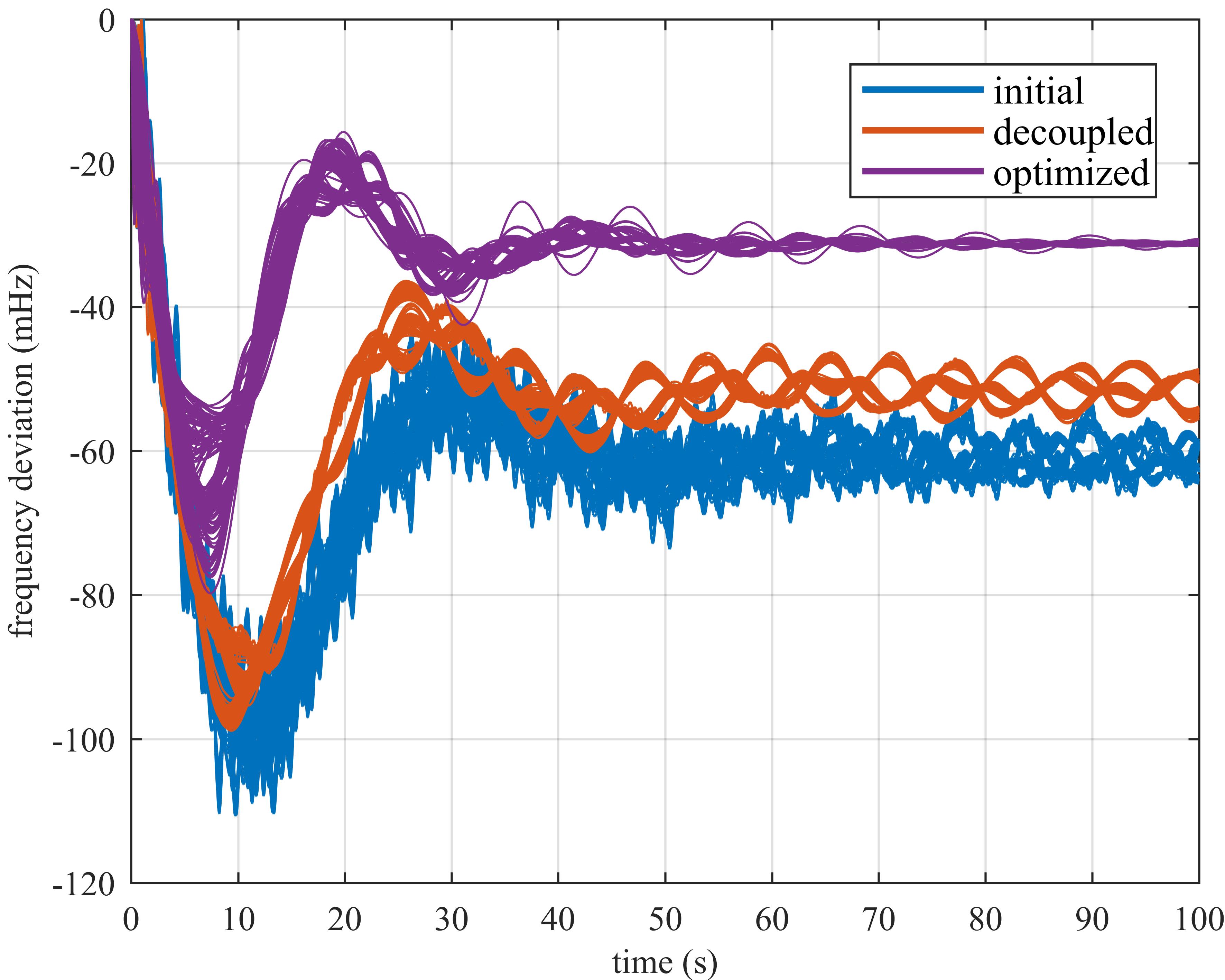}
	\caption{Frequency response of the synthetic power system to 100 MW load steps in 10 buses with the initial, decoupled, and optimized parameterization. Due to the large amount of dynamic prosumers, each parameterization is represented with one color.}
	\label{fig.BigSysStepResponse}
\end{figure}

For the hierarchical tuning, we replace each subsystem $\cS_i$ with a single-generator equivalent model in the model reduction step, analogously as in the previous example. 
Thus, in the optimization step, a system with 171 states and 90 controller parameters is optimized, instead of the original 2520 states and 1440 parameters, which is impossible to handle.

Figure~\ref{fig.BisSys_sigma} shows the singular value plot of the system with the initial parameterization and subsequent results. The largest peak, corresponding to the value 360 at approximately 1 rad/s, is cut-off to improve visibility of other plots.
With the proposed approach, the $\Hinf$ norm of the system was reduced by 95\% in four iterations. Thereby, the resonant peaks are almost eliminated. The decoupled parameterization, i.e. when each SO tunes its respective controller parameters separately, is also shown Fig.~\ref{fig.BisSys_sigma}. Even though the $\Hinf$ norm of the system is significantly improved compared to the initial parameterization, a large peak still remains.

The time-domain response of the system to load steps in 10 buses with the initial, decoupled, and optimized parameterization after four iterations, is shown in Fig.~\ref{fig.BigSysStepResponse}. With the decoupled parameterization, large oscillations still prevail at the end of the time horizon. On the other hand, the oscillations are significantly better dampened with the proposed approach. This shows that the presented approach is capable of optimizing such large and complex systems and can lead to very good results.

This results were achieved using a Windows computer with an Intel$^\circledR$ i7-4810MQ CPU running at 2.8 GHz. The model reduction step requires approx 30s on average for one subsystem. The optimization of the reduced system requires approx. 1 minute, whereas the model matching step requires less than 5 minutes for each subsystem due to the larger number of optimization parameters. Consequently, assuming that the model reduction and model matching steps for each subsystem are executed in parallel, approx. 7 minutes are needed for one iteration of the algorithm. Thus, even such large systems can be optimized in less than 30 minutes, depending on the number of iterations until a satisfactory system norm is achieved. Note that further subsystems can be coupled without increasing the computation time significantly, as the model reduction and model matching step are done in parallel.


\section{Conclusions and outlook}
\label{sec.Conclusion}

Tuning of controller parameters in power systems increases system resiliency when power system dynamics are constantly changing, e.g. due to an increasing share of renewable generation. 
The tuning, however, can be particularly challenging for two reasons: (a) the power systems can be very large, and (b) large power systems often belong to a multitude of subsystem operators which are not willing to exchange detailed information about their subsystems.
We proposed an algorithm for hierarchical data privacy conserving structured $\Hinf$ controller synthesis for power systems. The approach addresses both of these challenges it demonstrated its effectiveness with two simulation examples. 
It is based on the exchange of structured reduced models of subsystems, which conserves data privacy and reduces computational complexity. For this purpose, methods for model reduction and model matching are proposed.
In the first simulation example, we demonstrated how increased percentage of renewable generation leads to changing dynamics in the system, showing the need for online reparameterization of controllers. Additionally, we showed in the example that weakly-dampened oscillatory modes in the overall system cannot be eliminated if all subsystems optimize their parameters separately. The proposed hierarchical approach eliminates the weakly dampened modes in the first numerical example, and it achieves approximately the same results as with centralized tuning. In the second example, we showed the scalability and efficacy of the approach on an even larger power system, where centralized tuning cannot be applied, thereby practically eliminating oscillations from the system.


\appendix

\section{Controller models used for the IEEE 68 bus example}
\label{App.IEEE68Models}

\begin{figure}[t]
	\centering
	\includegraphics[width= 1\columnwidth]{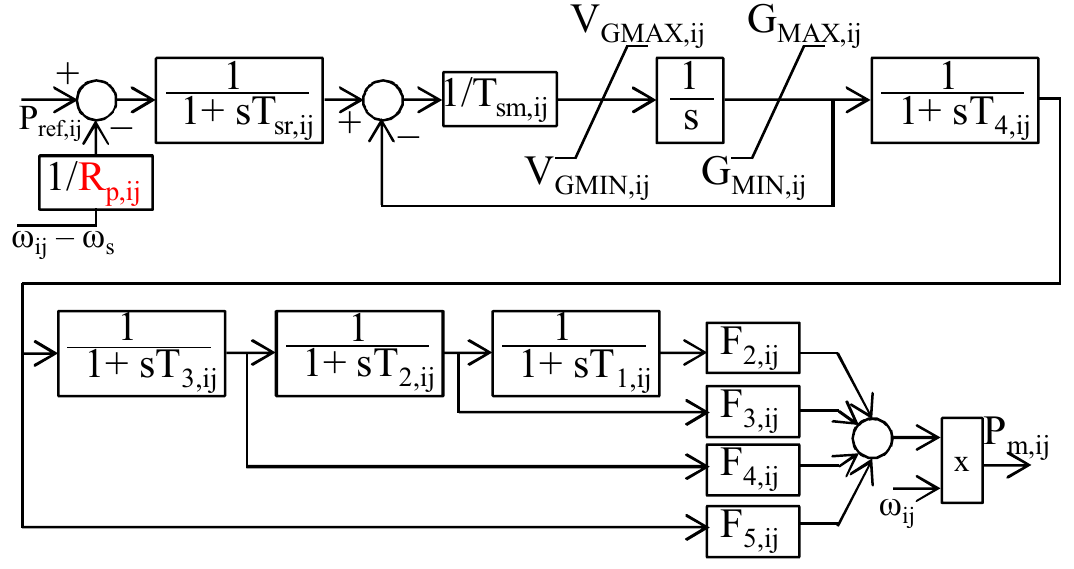}
	\caption{Dynamic model of the turbine and governor from~\cite{TGOVMathworks}. The frequency droop gain of the governor $R_{p,ij}$ is an optimization variable.}
	\label{fig.TGOV}
\end{figure}
\begin{figure}[t]
	\centering
	\includegraphics[width= 0.7\columnwidth]{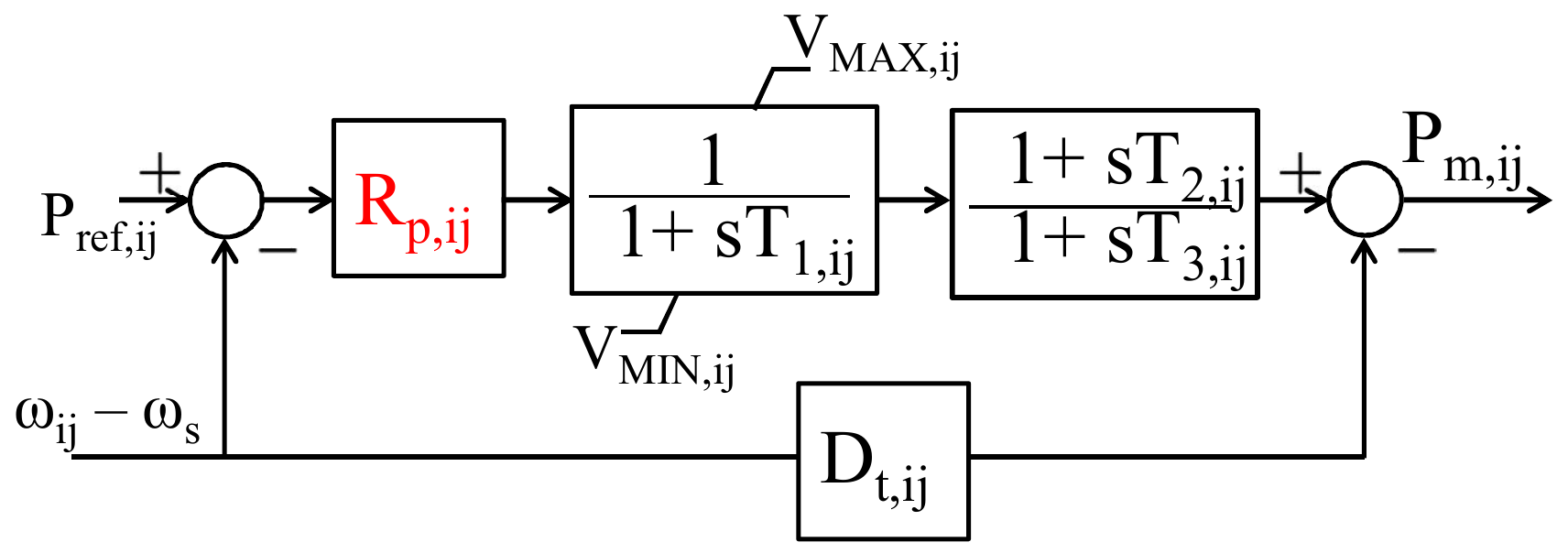}
	\caption{The standard TGOV1 turbine and governor model used for the power system model. The frequency droop gain of the governor $R_{p,ij}$ is an optimization variable.}
	\label{fig.TGOV1Model}
\end{figure}
\begin{figure}[tb]
	\centering
	\includegraphics[width=0.8\columnwidth]{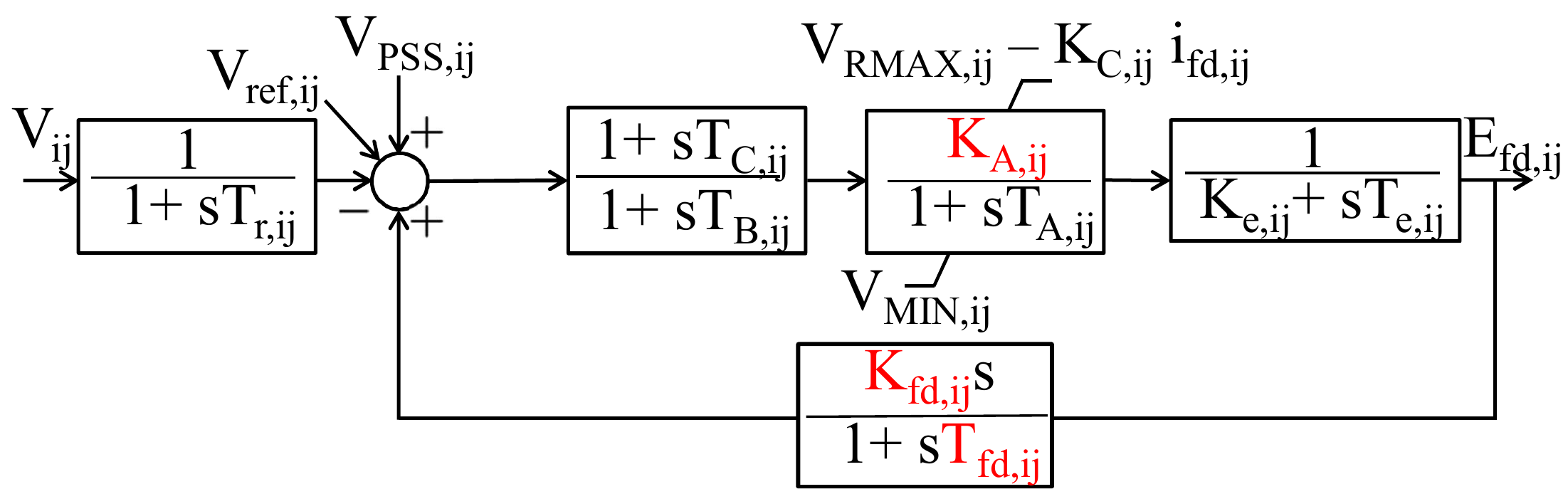}
	\caption{Dynamic model of AVR$_{ij}$~\cite{MathworksExciter}, where $T_{r,ij}$ is the transducer time constant, $T_{C,ij}$ and $T_{B,ij}$ are dynamic gain reduction time constants, $K_{A,ij}$ is the AVR gain, $T_{A,ij}$ is the AVR lag time constant, $K_{e,ij}$ and $T_{e,ij}$ are the exciter parameters, and $K_{fd,ij}$ and $T_{fd,ij}$ additional damping coefficients of the AVR. We consider $K_{A,ij}$, $K_{fd,ij}$, and $T_{fd,ij}$, marked red, as tunable parameters.}
	\label{fig.Exciter}
\end{figure}

\begin{figure}[tb]
	\centering
	\includegraphics[width=0.9\columnwidth]{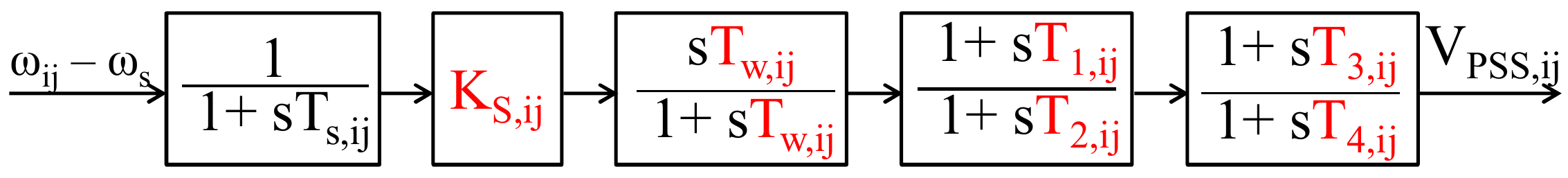}
	\caption{Dynamic model of the simple power system stabilizer (taken from~\cite{moeini2015open,kundur93a}), where $K_{S,ij}$ is the PSS gain, $T_{w,ij}$ is the washout time constant, $T_{1,ij}$-$T_{4,ij}$ are the lead-lag filters time constants, and $T_{s,ij}$ is the sensor time constant. All of the PSS parameters are tunable, except the sensor time constant.}
	\label{fig.PSS}
\end{figure}

Parameters of synchronous generators and of the power grid are provided in~\cite{singh2013ieee}, whereas Figures~\ref{fig.TGOV},~\ref{fig.TGOV1Model},~\ref{fig.Exciter}, and~\ref{fig.PSS} show the power plant controller models used for modeling of the IEEE 68 bus grid.
The governor and turbine models, shown in Figs.~\ref{fig.TGOV} and~\ref{fig.TGOV1Model} have one optimization parameter each, marked in red. It is the proportional gain of the governor. Thereby, half of the power plants have the model in Fig.~\ref{fig.TGOV1Model}, whereas other power plants have the model in Fig.~\ref{fig.TGOV}.
We optimize the gain $K_{A,ij}$ of the AVR$_{ij}$, shown red in Fig.~\ref{fig.Exciter}. We also optimize all parameters of PSS$_{ij}$ marked red in Fig.~\ref{fig.PSS}, except the physically-determined sensor time constant.
All presented controller models are standard IEEE models.
\end{document}